\documentclass[journal]{IEEEtran}

\usepackage[english]{babel}

\usepackage{ifpdf}

\usepackage{cite} 
\usepackage{url}
\usepackage{hyperref}

\ifCLASSINFOpdf
	\usepackage[pdftex]{graphicx}
	\graphicspath{{./figures/}}
\else
	\usepackage[dvips]{graphicx}
	\graphicspath{./figures/}
\fi
\usepackage{color}
\usepackage{pgf, tikz, pgfplots}
\usetikzlibrary{shapes, arrows, automata}

\usepackage[cmex10]{amsmath}
\usepackage{amsfonts, amssymb, amsthm}
\usepackage{mathrsfs}

\usepackage[]{algorithm2e}

\usepackage{enumerate}
\usepackage{multirow}
\usepackage{rotating}
\usepackage{subcaption}
\usepackage{needspace}

\input{mySymbol.sty}

\definecolor{penndarkestblue}{cmyk}{1,0.74,0,0.77}
\definecolor{penndarkerblue}{cmyk}{1,0.74,0,0.70}
\definecolor{pennblue}{cmyk}{0.99,0.66,0,0.57} 
\definecolor{pennlighterblue}{cmyk}{0.98,0.44,0,0.35}
\definecolor{pennlightestblue}{cmyk}{0.38,0.17,0,0.17} 

\definecolor{penndarkestred}{cmyk}{0,1,0.89,0.66}
\definecolor{penndarkerred}{cmyk}{0,1,0.88,0.55}
\definecolor{pennred}{cmyk}{0,1,0.83,0.42} 
\definecolor{pennlighterred}{cmyk}{0,1,0.6,0.24}
\definecolor{pennlightestred}{cmyk}{0,0.43,0.26,0.12} 

\definecolor{penndarkestgreen}{cmyk}{1,0,1,0.68}
\definecolor{penndarkergreen}{cmyk}{1,0,1,0.57}
\definecolor{penngreen}{cmyk}{1,0,1,0.44} 
\definecolor{pennlightergreen}{cmyk}{1,0,1,0.25}
\definecolor{pennlightestgreen}{cmyk}{0.43,0,0.43,0.13}

\definecolor{penndarkestorange}{cmyk}{0,0.65,1,0.49}
\definecolor{penndarkerorange}{cmyk}{0,0.65,1,0.33}
\definecolor{pennorange}{cmyk}{0,0.54,1,0.24} 
\definecolor{pennlighterorange}{cmyk}{0,0.32,1,0.13}
\definecolor{pennlightestorange}{cmyk}{0,0.15,0.46,0.06}
	
\definecolor{penndarkestpurple}{cmyk}{0,1,0.11,0.86}
\definecolor{penndarkerpurple}{cmyk}{0,1,0.13,0.82}
\definecolor{pennpurple}{cmyk}{0,1,0.11,0.71} 
\definecolor{pennlighterpurple}{cmyk}{0,1,0.05,0.46}
\definecolor{pennlightestpurple}{cmyk}{0,0.35,0.02,0.23}
	
\definecolor{pennyellow}{cmyk}{0,0.20,1,0.05} 
\definecolor{pennlightgray1}{cmyk}{0,0,0,0.05}
\definecolor{pennlightgray3}{cmyk}{0.01,0.01,0,0.18}
\definecolor{pennmediumgray1}{cmyk}{0.04,0.03,0,0.31}
\definecolor{pennmediumgray4}{cmyk}{0.08,0.06,0,0.54}
\definecolor{penndarkgray2}{cmyk}{0.09,0.07,0,0.71}
\definecolor{penndarkgray4}{cmyk}{0.1,0.1,0,0.92}

\def\Tr{\mathsf{T}}
\def\Hr{\mathsf{H}}

\def\jj{\mathfrak{j}}
\def\dc{\mathrm{dc}}
\def\ER{\mathrm{ER}}

\newtheorem{lemma}{\hspace{0pt}\bf Lemma}
\newtheorem{proposition}{\hspace{0pt}\bf Proposition}

\newtheorem{theorem}{\hspace{0pt}\bf Theorem}
\newtheorem{corollary}{\hspace{0pt}\bf Corollary}

\newtheorem{remark}{\hspace{0pt}\bf Remark}

\newtheorem{definition}{\hspace{0pt}\bf Definition}

\def \mod {\text{\,mod\,}}

\begin{document}

\title{Ergodicity in Stationary Graph Processes: \\ A Weak Law of Large Numbers}

\author{Fer\hspace{0.02cm}nando~Gama~
        and~Alejandro~Ribeiro
\thanks{Supported by NSF CCF 1717120, ARO W911NF1710438, ARL DCIST CRA W911NF-17-2-0181, ISTC-WAS and Intel DevCloud. Authors are with Dept. of Electrical and Systems Eng., Univ. of Pennsylvania, \{fgama, aribeiro\}@seas.upenn.edu. Part of the results in this paper appeared in \cite{gama17}.}
}

\markboth{IEEE TRANSACTIONS ON SIGNAL PROCESSING (ACCEPTED)}%
{Ergodicity in Stationary Graph Processes}

\maketitle

\begin{abstract}
For stationary signals in time the weak law of large numbers (WLLN) states that ensemble and realization averages are within $\epsilon$ of each other with a probability of order $\ccalO(1/N\eps^2)$ when considering $N$ signal components. The graph WLLN introduced in this paper shows that the same is essentially true for signals supported on graphs. However, the notions of stationarity, ensemble mean, and realization mean are different. Recent papers have defined graph stationary signals as those that satisfy a form of invariance with respect to graph diffusion. The ensemble mean of a graph stationary signal is not a constant but a node-varying signal whose structure depends on the spectral properties of the graph. The realization average of a graph signal is defined here as an average of successive weighted averages of local signal values with signal values of neighboring nodes. The graph WLLN shows that these two node-varying signals are within $\eps$ of each other with probability of order $\ccalO(1/N\eps^2)$ in at least some nodes. In stationary time signals, the realization average is not only a consistent estimator of the ensemble mean but also optimal in terms of mean squared error (MSE). This is not true of graph signals. Optimal MSE graph filter designs are also presented. An example problem concerning the estimation of the mean of a Gaussian random field is presented.
\end{abstract}

\begin{IEEEkeywords}
Graph signal processing, ergodicity, law of large numbers, unbiased, consistent, optimal estimators
\end{IEEEkeywords}

\IEEEpeerreviewmaketitle


\section{Introduction} \label{sec_intro}

Random signals and stochastic processes provide the foundation of statistical signal processing which is concerned with the problem of extracting useful information out of random (noisy) data. One important concept in the field is the notion of stationarity that can be defined on signals with regular structure such as images or signals in time \cite{doob53, shanmugan88}. In particular, wide sense stationarity (WSS) of time signals models processes which have constant mean and a correlation function that only depends on the time elapsed between different signal components. A fundamental property that arises in WSS is ergodicity: The equivalence between the {\it ensemble} and the {\it realization} mean \cite{petersen83, gray09}. This property is valuable in situations where we have access to a single realization of the process because it allows estimation of the {\it ensemble} mean -- a property of the process -- using the {\it realization} mean -- a property of the individual given realization. Ergodicity results are a manifestation of the Law of Large Numbers (LLN) as they state convergence of the sample mean obtained from averaging all the time samples of a single realization to the true ensemble mean of the process. The Weak (W)LLN states that this convergence is obtained in probability \cite[Ch. 3]{sen93}. The pointwise ergodic theorem proves almost sure convergence for stationary processes \cite[Thm. 7.2.1]{durrett10}.

\begin{figure*}[t!]
	\centering
	\begin{subfigure}{0.31\textwidth}
		\centering
		\includegraphics[width=\textwidth]{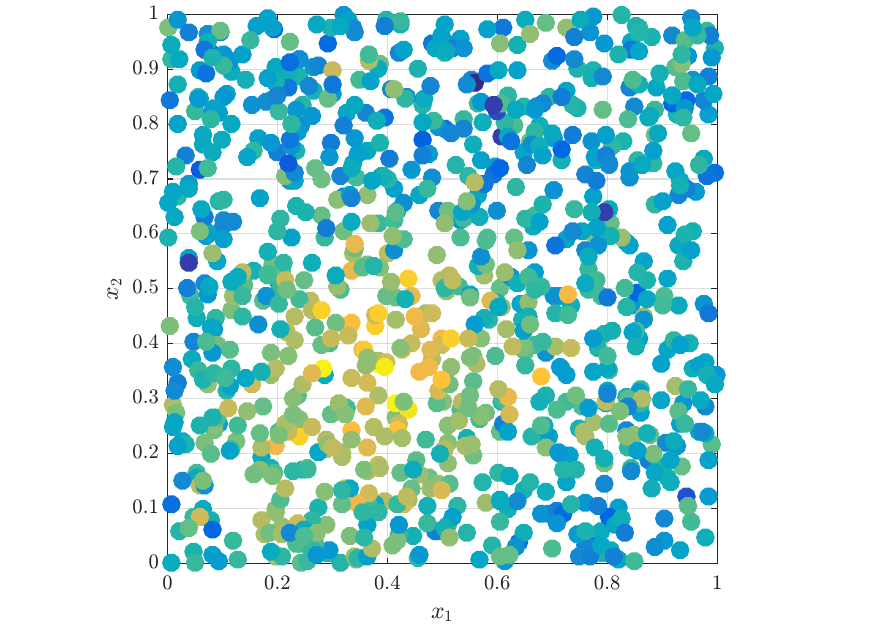}
		\caption{Single realization}
		\label{fig_sensor_measurements}
	\end{subfigure}
	\hfill
	\begin{subfigure}{0.31\textwidth}
		\centering
		\includegraphics[width=\textwidth]{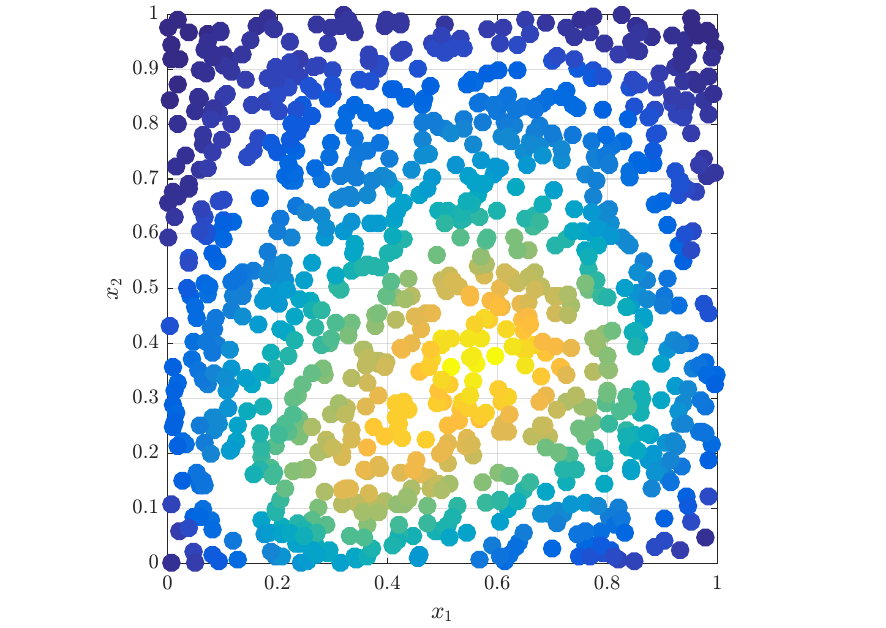}
		\caption{Graph \emph{realization} mean}
		\label{fig_sample_mean}
	\end{subfigure}
	\hfill
	\begin{subfigure}{0.31\textwidth}
		\centering
		\includegraphics[width=\textwidth]{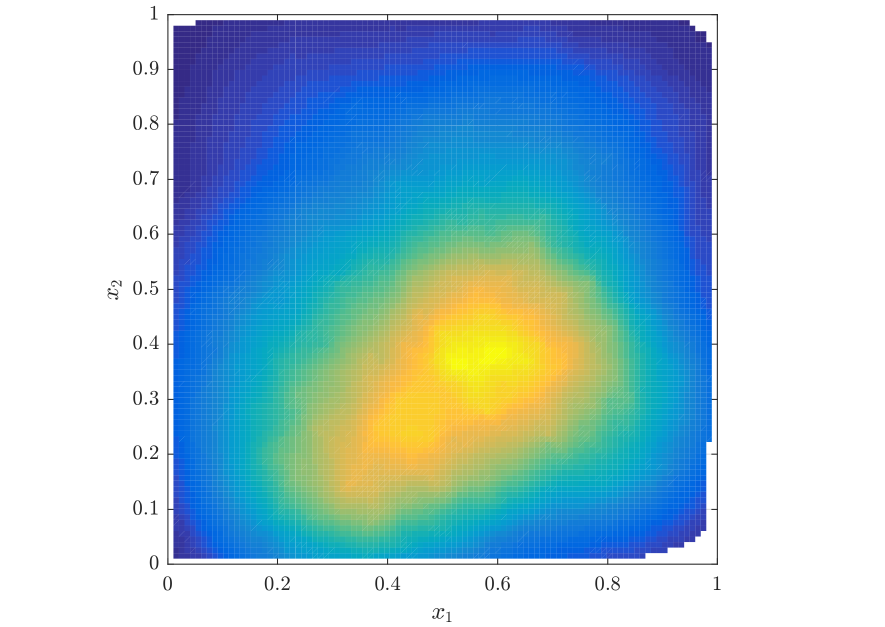}
		\caption{\emph{Ensemble} mean}
		\label{fig_true_mean}
	\end{subfigure}
	\caption{Illustration of a Gaussian Markov Random Field (GMRF) being measured by a sensor network. \subref{fig_sensor_measurements} Measurements of the GMRF (color-coded) taken by the deployed sensors. \subref{fig_sample_mean} Values at each sensor obtained after computing the graph \emph{realization} mean (defined as the \emph{graph shift average} described in Section~\ref{sec_ergodicity}). \subref{fig_true_mean} True \emph{ensemble} mean of the GMRF. We observe that the graph \emph{realization} mean is close to the true \emph{ensemble} mean. We also note that the ensemble mean is not a constant vector.}
	\label{fig_gmrf_1000}
\end{figure*}

The goal of this paper is to generalize the WLLN to signals supported in irregular domains described by arbitrary graphs \cite{dimakis10, hammond11}. We do so by building on the theory of graph signal processing (GSP) \cite{sandryhaila13, sandryhaila14, shuman13, sandryhaila14-mag} and, more specifically, on the concept of WSS graph signals \cite{girault15, perraudin17, marques17}. GSP is built on the notion of a graph shift operator. This is a matrix representation of the graph that defines a linear operator that can be applied to graph signals. The graph shift diffuses the signal through the graph and it represents a local operation that permits comparison of the value of the signal with the values of neighboring nodes. This is a natural generalization of a time shift, which compares signal values at adjacent points in time -- and, indeed, the time shift is a particular case of a graph shift. This motivates definition of the graph Fourier transform (GFT) as a projection on the eigenvector space of the graph shift. Since eigenvectors can be shown to represent different modes of variability, the GFT provides a natural description of (graph) frequency components. The relationship between time and graph shifts is also exploited in the definition of graph WSS signals which are required to satisfy a form of invariance with respect to the graph shift operator \cite{marques17}. As is the case of time signals, this invariance induces conditions in the first and second moments of the signal. The ensemble mean of a WSS graph signal must be aligned with an eigenvector of the graph shift operator and the covariance matrix is diagonalized by its eigenvector basis. As a consequence of the latter, WSS graph signals have a covariance structure that is completely characterized by a power spectral density (PSD).

The question we ask in this paper is whether an analogous notion of ergodicity can be defined for WSS graph signals. We will provide an affirmative answer but in understanding our results it is important to explain the differences between the {\it ensemble} and {\it realization} mean for graph signals and their respective counterparts for time signals. For example, when we have a Markov random field observed by a sensor network it is ready to construct a graph linking the sensors so that the sensor's observations are WSS with respect to the graph. Fig.~\ref{fig_gmrf_1000}-(c) provides an illustration of the ensemble mean of the Markov random field monitored by the sensor network, Fig.~\ref{fig_gmrf_1000}-(a) measurements of a single realization, and Fig.~\ref{fig_gmrf_1000}-(b) the graph realization mean of this set of observations. It is apparent that the ensemble mean of a graph signal is not a constant and that, consequently, the graph realization mean is not a simple average of the values observed by individual nodes. Still, the similarity between Fig.~\ref{fig_gmrf_1000}-(b) and Fig.~\ref{fig_gmrf_1000}-(c) illustrates that it is possible to define an {\it ensemble} mean for graph stationary signals (Def. \ref{def_stationarity}) and to define a graph {\it realization} mean (Def. \ref{def_graph_shift_average}) that is close to the true ensemble mean.

We begin the paper by introducing necessary GSP concepts to define the graph realization mean and graph WSS signals (Sec. \ref{sec_prelim}). To emphasize the difference between the graph realization mean and the sample mean, we refer to the former as {\it graph shift averages.} A graph shift average considers subsequent applications of the graph shift and averages the resulting values (Def.~\ref{def_graph_shift_average}). From the perspective of a given node, this is tantamount to computing a weighted average of its own information with the information of neighbors, followed by a weighted average with neighboring weighted averages, and so on. Eventually, each node utilizes network-wide information to refine an estimate of the mean at its own location while relying on interactions with neighboring nodes only -- this graph shift average is used to produce the signal in Fig.~\ref{fig_gmrf_1000}-(b) from the observations in Fig.~\ref{fig_gmrf_1000}-(a).

To elucidate ergodicity properties (Sec. \ref{sec_ergodicity}) we study the moments of graph shift averages (Sec. \ref{subsec_diffusion}). We show that {\it graph shift averages} are unbiased estimators of the {\it graph ensemble mean} (Prop. \ref{prop_unbiased}) and that they are themselves stationary with respect to the graph shift operator (Prop. \ref{prop_psd_graph_shift_average}). The latter fact allows to show that the graph shift average behaves like a low pass filter that keeps the DC graph frequency component unchanged but attenuates all other coefficients. From this spectral property we derive the graph weak law of large numbers which states that the difference between the ensemble mean and the graph shift average exceeds $\epsilon$ with a probability of order $\ccalO(1/N\eps^2)$ for a graph with $N$ nodes (Thm. \ref{thm_glln}). This is the same property that can be established from signals in time. We observe that in the latter case it is possible to let $N$ grow to conclude convergence in probability; a result that we can recover as a particular case (Cor. \ref{coro_wlln}). In the general case of graph signals, we cannot let $N$ grow unless we assume the underlying graph support belongs to a certain class of graphs for which the limit of $N$ has proper meaning \cite{wolfe13, segarra17-graphons}. Therefore, on the one hand, we provide results that depend on the size $N$ of the graph and establish conditions on how the spectral properties of the graph should behave with increasing $N$, and on the other hand, we study cases on which this behavior is satisfied, such as random Erd\H{o}s-R{\'e}nyi graphs, in order to exemplify how convergence in probability results can be obtained from our probability bounds for finite $N$ (Cor. \ref{coro_er}).

In time signals the realization mean is also an optimal estimator in the sense of minimizing the mean squared error (MSE). In the case of graph signals this is not true. The MSE of graph shift averages can be further improved by adequately rescaling each subsequent application of the graph shift; i.e., through the use of a graph filter (Sec.~\ref{sec_optimal}). Not surprisingly, the filter that minimizes the MSE is an ideal low pass filter (Thm.~\ref{thm_optimal_MSE}) which we further show is also optimal in terms of minimizing the volume of the error ellipsoid (Thm.~\ref{thm_optimal_logdet}). Numerical experiments on the application of graph shift averages are conducted on Erd\H{o}s-R\'{e}nyi graphs, covariance graphs and stochastic block models to corroborate convergence to the true ensemble mean as the number of nodes increases (Sec.~\ref{sec_sims}). An example application consisting of the aforementioned problem of estimating the mean of a Markov random field using a single realization of the field and relying on local information exchanges is presented and compared to existing approaches \cite{dilorenzo14}. We close the paper with concluding remarks (Sec.~\ref{sec_conclusions}).


%
\section{Graph Shift Averages} \label{sec_prelim}

Let $\ccalG=(\ccalV,\ccalE,\ccalW)$ be a connected graph with $N$ vertices $n\in\ccalV$, edges $e\in\ccalE \subseteq \ccalV \times \ccalV$, and weight function $\ccalW:\ccalE \to \reals$. Further define the graph shift operator as a matrix $\bbS \in \reals^{N \times N}$ having the same sparsity pattern of the graph $\ccalG$ so that we can have $[\bbS]_{ij} = s_{ij}\neq0$ only when $(j,i)\in\ccalE$. The graph shift operator is a stand in for any of the various matrix representations associated with the graph $\ccalG$ \cite{sandryhaila13}.  Particular cases include the adjacency matrix $\bbA$ with nonzero elements $[\bbA]_{ij}=a_{ij}=\ccalW(j,i)$ for all $(j,i)\in \ccalE$, the diagonal degree matrix $\bbD=\diag(\bbA\bbone)$, and the normalized adjacency $\bbA_{\text{norm}} = \bbD^{-1/2}\bbA\bbD^{-1/2}$. 

We assume that the graph shift operator is normal, which implies existence of an orthonormal basis of eigenvectors $\bbV = [\bbv_{1},\ldots,\bbv_{N}] \in \mbC^{N \times N}$ satisfying $\bbV\bbV^{\Hr}=\bbI$ and a diagonal eigenvalue matrix $\bbLambda = \diag(\lambda_{1},\ldots,\lambda_{N}) \in \mbC^{N \times N}$ such that we can write the shift operator as $\bbS = \bbV \bbLambda \bbV^{\Hr}$. Adjacency matrices of undirected graphs are symmetric and therefore always normal. We further restrict attention to nonnegative operators satisfying $s_{ij}\geq0$ for all $i,j$. For these operators the Perron-Frobenius Theorem applies and we can guarantee that a unique real nonnegative eigenvalue attains maximal norm \cite{gantmacher89}. Without loss of generality we assume that $\lambda_1$ is such eigenvalue so that we have
\begin{equation} \label{eqn_perron_frobenius}
   |\lambda_n| \leq \lambda_1 \in\reals^+, \quad \text{for all } n\neq1 .
\end{equation}
Although we may have eigenvalues with $|\lambda_n| = \lambda_1$ we will relax nomenclature to say that $\lambda_1$ is the largest eigenvalue of $\bbS$. The eigenvector $\bbv_1=[v_{11},\ldots, v_{N1}]^{\Tr} \in \mbC^{N}$ associated with eigenvalue $\lambda_1$ can be chosen to be real and nonnegative and it plays an important role in the definition of graph stationary signals (Section \ref{sec_graph_stationarity}) and the ergodicity theorems in Sections \ref{sec_ergodicity} and \ref{sec_optimal}. The adjacency and normalized adjacency are nonnegative shift operators but the Laplacian $\bbL = \bbD-\bbA$ does not belong to this category -- see Remark \ref{rmk_laplacian}. 
 
Associated with the nodes of the graph we have a random signal $\bbx=[x_{1},\ldots, x_{N}]^{\Tr}\in\reals^{N}$. We use $\bbmu:=\mbE[\bbx]$ to denote the mean of the signal $\bbx$ which represents the {\it ensemble mean.} We use $\bbC_{x}:=\mbE[(\bbx-\bbmu)(\bbx-\bbmu)^{\Hr}]$ to represent the covariance matrix of the random graph signal $\bbx$. The signal $\bbx$ is defined on an irregular domain and the graph shift $\bbS$ is a description of the underlying signal structure. The structure of the graph induces a domain-specific transform that we call the graph Fourier transform (GFT) of $\bbx$ associated with the shift $\bbS$ \cite{sandryhaila14}. The GFT of $\bbx$ is the signal $\tbx$ obtained after projecting on the eigenvector basis of $\bbS$ and the inverse GFT (iGFT) transform is defined as the inverse projection operation,
\begin{equation} \label{eqn_gft_def}
   \tbx = \bbV^{\Hr}\bbx  \quad \Longleftrightarrow \quad
   \bbx = \bbV \tbx.
\end{equation}
The elements $[\tbx]_{n}=\tdx_{n}$ of $\tbx$ are known as the frequency coefficients of the signal and the eigenvectors $\bbv_{n}$ are the frequency components or frequency basis. The GFT is also a random signal with mean $\tbmu:=\mbE[\tbx]=\bbV^{\Hr}\bbmu$ and covariance matrix $\bbC_{\tdx}:=\mbE[(\tbx-\E{\tbx})(\tbx-\E{\tbx})^{\Hr}]$. To relate $\bbC_{\tdx}$ to the covariance matrix $\bbC_{x}$ observe that it follows from the definition of the GFT in \eqref{eqn_gft_def} that $\bbC_{\tdx}=\bbV^{\Hr}\mbE[(\bbx-\E{\bbx})(\bbx-\E{\bbx})^{\Hr}]\bbV$. The expectation in this latter expression is the covariance $\bbC_{x}$ of the signal in the node domain implying that $\bbC_{\tdx} = \bbV^{\Hr}\bbC_x\bbV$.

Of particular importance in our forthcoming discussions is the cyclic graph $\ccalG_{\dc}$ with edge set composed of all the edges of the form $(n, 1 + n \mod N)$ for $n\in[1,N]$; see Figure \ref{fig_dc}. Let $\bbA_{\dc}$ denote the corresponding adjacency matrix with nonzero elements $[\bbA_{\dc}]_{1 + n \mod N, n}=1$. The cyclic graph is a natural description of discrete (periodic) time and we can therefore say that a discrete time signal $\bbx$ is a graph signal defined on the cyclic graph $\ccalG_{\dc}$. Since the cycle is a circulant graph it is diagonalized by the Fourier basis and it follows that the GFT reduces to the conventional DFT for this particular graph.

%
\begin{figure}
	\centering

\def \thisplotscale {0.55}
\def \unit {\thisplotscale cm}

\def \xdisplaced{5}
\def \dotsize{7}

\tikzstyle{dot} = 
	[ellipse,
	 inner sep = 0pt,
	 anchor = center,
	 draw=black,
	 minimum height = \dotsize,
	 minimum width  = \dotsize]

\tikzstyle{cover} = 
	[ellipse,
	 inner sep = 0pt,
	 anchor = center]

\def\ypos{1.5}
\def\xpos{1.5}
\def\lbl{0.75}
\def\mylinewidth{0.75}

{\footnotesize
\begin{tikzpicture}[scale = \thisplotscale]

	\node at (0,0) (DC1) {};
	\node at (\xdisplaced,0) (DC2) {};
	\node at (2*\xdisplaced,0) (DC3) {};
	
	
	\path (0,\ypos) node[dot,fill=pennblue] (x11) {} ++ (0,\lbl) node {${
	 x_{1}}$};
	\path (0.866*\xpos,0.5*\ypos) node[dot,fill=penngreen] (x12) {} ++ (0.866*\lbl,0*\lbl) node {${
	 x_{2}}$};
	\path (0.866*\xpos,-0.5*\ypos) node[dot,fill=pennyellow] (x13) {} ++ (0.866*\lbl,-0*\lbl) node {${
	 x_{3}}$};
	\path (0,-1*\ypos) node[dot,fill=pennorange] (x14) {} ++ (0,-1*\lbl) node {${
	 x_{4}}$};
	\path (-0.866*\xpos,-0.5*\ypos) node[dot,fill=pennred] (x15) {} ++ (-0.866*\lbl,-0*\lbl) node {${
	 x_{5}}$};
	\path (-0.866*\xpos,0.5*\ypos) node[dot,fill=pennpurple] (x16) {} ++ (-0.866*\lbl,0*\lbl) node {${
	 x_{6}}$};
	
	\path (x11) edge[bend left=20,-stealth, line width=\mylinewidth] (x12);
	\path (x12) edge[bend left=20,-stealth, line width=\mylinewidth] (x13);
	\path (x13) edge[bend left=20,-stealth, line width=\mylinewidth] (x14);
	\path (x14) edge[bend left=20,-stealth, line width=\mylinewidth] (x15);
	\path (x15) edge[bend left=20,-stealth, line width=\mylinewidth] (x16);
	\path (x16) edge[bend left=20,-stealth, line width=\mylinewidth] (x11);
	
	
	\path (\xdisplaced,0) ++ (0,\ypos) node[dot,fill=pennpurple] (x21) {} ++ (0,\lbl) node {${x_{1}}+{x_{6}}$};
	\path (\xdisplaced,0) ++(0.866*\xpos,0.5*\ypos) node[dot,fill=pennblue] (x22) {} ++ (0.866*\lbl,1*\lbl) node {${x_{2}}+{ x_{1}}$};
	\path (\xdisplaced,0) ++(0.866*\xpos,-0.5*\ypos) node[dot,fill=penngreen] (x23) {} ++ (0.866*\lbl,-1*\lbl) node {${ x_{3}}+{ x_{2}}$};
	\path (\xdisplaced,0) ++(0,-1*\ypos) node[dot,fill=pennyellow] (x24) {} ++ (0,-1*\lbl) node {${ x_{4}}+{ x_{3}}$};
	\path (\xdisplaced,0) ++(-0.866*\xpos,-0.5*\ypos) node[dot,fill=pennorange] (x25) {} ++ (-0.866*\lbl,-1*\lbl) node {${ x_{5}}+{ x_{4}}$};
	\path (\xdisplaced,0) ++(-0.866*\xpos,0.5*\ypos) node[dot,fill=pennred] (x26) {} ++ (-0.866*\lbl,1*\lbl) node {${ x_{6}}+{ x_{5}}$};
	
	\path (x21) edge[bend left=20,-stealth, line width=\mylinewidth] (x22);
	\path (x22) edge[bend left=20,-stealth, line width=\mylinewidth] (x23);
	\path (x23) edge[bend left=20,-stealth, line width=\mylinewidth] (x24);
	\path (x24) edge[bend left=20,-stealth, line width=\mylinewidth] (x25);
	\path (x25) edge[bend left=20,-stealth, line width=\mylinewidth] (x26);
	\path (x26) edge[bend left=20,-stealth, line width=\mylinewidth] (x21);
	
	
	\draw (DC1.center) ++ (0.866*\xpos,2*\ypos) edge[bend left=20,-stealth, line width=2*\mylinewidth] node[midway,above] {$\bbx+\bbA_{\dc}\bbx$}(-0.866*\xpos+1*\xdisplaced,2*\ypos);
	
	
	\path (2.2*\xdisplaced,0) ++ (0,\ypos) node[dot,fill=penngreen] (x31) {} ++ (0,\lbl) node {$\sum x_{k}$};
	\path (2.2*\xdisplaced,0) ++(0.866*\xpos,0.5*\ypos) node[dot,fill=pennyellow] (x32) {} ++ (0.866*\lbl,0.8*\lbl) node {$\sum x_{k}$};
	\path (2.2*\xdisplaced,0) ++(0.866*\xpos,-0.5*\ypos) node[dot,fill=pennorange] (x33) {} ++ (0.866*\lbl,-0.8*\lbl) node {$\sum x_{k}$};
	\path (2.2*\xdisplaced,0) ++(0,-1*\ypos) node[dot,fill=pennred] (x34) {} ++ (0,-1*\lbl) node {$\sum x_{k}$};
	\path (2.2*\xdisplaced,0) ++(-0.866*\xpos,-0.5*\ypos) node[dot,fill=pennpurple] (x35) {} ++ (-0.866*\lbl,-0.8*\lbl) node {$\sum x_{k}$};
	\path (2.2*\xdisplaced,0) ++(-0.866*\xpos,0.5*\ypos) node[dot,fill=pennblue] (x36) {} ++ (-0.866*\lbl,0.8*\lbl) node {$\sum x_{k}$};
	
	\path (x31) edge[bend left=20,-stealth, line width=\mylinewidth] (x32);
	\path (x32) edge[bend left=20,-stealth, line width=\mylinewidth] (x33);
	\path (x33) edge[bend left=20,-stealth, line width=\mylinewidth] (x34);
	\path (x34) edge[bend left=20,-stealth, line width=\mylinewidth] (x35);
	\path (x35) edge[bend left=20,-stealth, line width=\mylinewidth] (x36);
	\path (x36) edge[bend left=20,-stealth, line width=\mylinewidth] (x31);
	
	
	\draw (DC2.center) ++ (0.866*\xpos,2*\ypos) edge[dashed, bend left=20,-stealth, line width=2*\mylinewidth] node[midway,above] {$\sum_{\ell=0}^{5} \bbA_{\dc}^{\ell} \bbx$}(-0.866*\xpos+2.2*\xdisplaced,2*\ypos);

\end{tikzpicture}
}
	\caption{Graph shift average in discrete-time processes. A discrete-time process can be described as supported by a directed cycle graph. Then, application of the graph shift $\bbS_{\dc}=\bbA_{\dc}$ \emph{moves} the value of the signal at one node to the next, generating the causality typical of time. This means that, by aggregating all shifted version of the signal (i.e. computing the graph shift average) at a single node, one can construct the arithmetic mean at every node.}
	\label{fig_dc}
\end{figure}
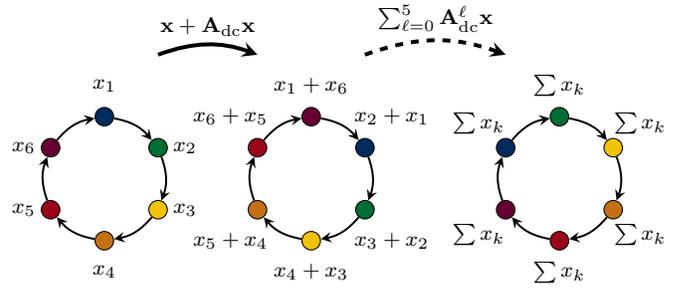

Our interest in this paper is to define and study graph shift averages of random graph signals. To motivate the definition recall that in the case of time signals the time average of the signal $\bbx$ is the arithmetic mean
\begin{equation} \label{eqn_time_average}
	\hhatmu_{N} =  \frac{1}{N} \sum_{n=1}^{N} x_{n} .
\end{equation}
If we think of $\bbx$ as a windowed version of a \emph{stationary} stochastic process, we know that as $N$ grows, the \emph{realization} average $\hhatmu_{N}$ converges to the \emph{ensemble} mean of the process $\mbE[x_{n}] = \mu$ \cite{gray09}. 

For general graph signals, notice that the shift operator $\bbS$ represents a local transformation that when applied to the graph signal $\bbx$ produces the signal $\bby=\bbS\bbx$ in which the component $[\bby]_{i} = y_{i} = \sum_{(j,i)\in\ccalE} s_{ij} x_j$ depends on the values of $\bbx$ at nodes $j$ that are adjacent to $i$. The locality of the shift operator motivates the following definition of {\it graph shift averages}.

%
\begin{definition}\label{def_graph_shift_average}
Given a graph with shift operator $\bbS$ and a corresponding graph signal $\bbx$, the graph shift average of $\bbx$ is the graph signal
\begin{equation} \label{eqn_graph_shift}
	\hbmu_{N} := \frac{1}{\alpha(\bbS)} \sum_{\ell=0}^{N-1} \bbS^{\ell}\bbx,
\end{equation}
for some constant $\alpha(\bbS)>0$ that depends on the topology of the underlying graph support; see Proposition \ref{prop_unbiased}. \end{definition}

%
The graph shift average is a simple diffusion of the signal $\bbx$ through the graph. A practically important property of diffusions is that they can be implemented in a distributed manner. The $n$th element of $\hbmu_{N}$ can be obtained at node $n$ through $N-1$ communication exchanges with neighboring nodes.

When $\bbS = \bbA_{\dc}$, application of the shift operator results in the signal $\bby=\bbA_{\dc}\bbx$ with components ${y_{1+n\mod N}}=x_n$ implying that $\bby$ is a time shifted copy of $\bbx$. Making $\alpha(\bbS) = N$ and $\bbS=\bbA_{\dc}$ in \eqref{eqn_graph_shift}, it follows that $\hbmu_{N}=\hhatmu_{N}\bbone$. This justifies the use of \eqref{eqn_graph_shift} as a generalization of \eqref{eqn_time_average}. As we will see in Section \ref{sec_ergodicity}, the use of the {\it graph shift average} as a generalization of {\it realization averages} is further justified by its relationship to the {\it ensemble mean} $\bbmu$. This relationship will hold for graph stationary signals \cite{girault15, perraudin17, marques17} which we introduce in the following section.

%
\subsection{Wide Sense Stationarity of Graph Signals}\label{sec_graph_stationarity}

Wide sense stationarity (WSS) of graph signals is defined with respect to an associated graph shift operator $\bbS$ and imposes certain statistical structure on the random graph signal $\bbx$. We adopt the following definition \cite{marques17}.

%
\begin{definition}[WSS Graph Signal] \label{def_stationarity}
Let $\bbx \in \reals^{N}$ be a random graph signal supported on the $N$-node graph $\ccalG$ having associated graph shift $\bbS\in\reals^{N\times N}$ which we assume nonnegative and normal with eigenvalue  decomposition $\bbS = \bbV \bbLambda \bbV^{\Hr}$. The signal $\bbx$ is a WSS graph signal with respect to the graph shift $\bbS$ if 
\begin{enumerate}[(i)]
\item The expectation of $\bbx$ is proportional to the eigenvector of the graph shift $\bbS$ associated to its largest eigenvalue. I.e., there exist a constant $\mu \in \reals$ such that
\begin{equation}\label{eqn_mean_definition}
   \bbmu \ :=\ \mbE[\bbx]
         \  =\ \mu \bbv_{1} .
\end{equation}
\item The covariance matrix of $\bbx$ is diagonalized by the eigenvector basis of the shift operator. I.e., there exists a vector $\bbp \in \reals^{N}$ such that
\begin{equation}\label{eqn_joint_diagonalization}
   \bbC_{x}\ :=\ \mbE \Big[ (\bbx-\bbmu) (\bbx-\bbmu)^{\Hr} \Big] 
           \  =\ \bbV \diag(\bbp) \bbV^{\Hr} .
\end{equation}
We say that $\bbp$ is the power spectral density (PSD) of $\bbx$.
\end{enumerate}
\end{definition}

%
\noindent In aligning the mean with an eigenvector of $\bbS$, Property (i) in Def.~\ref{def_stationarity} implies that the direction of the mean remains unchanged after successive applications of the shift operator. This is an intuitively reasonable invariance that is consistent with stationarity of time signals because the Perron-Frobenius root of $\bbA_{\dc}$ is $\bbv_1=\bbone$. Further observe that because $\bbmu = \mu \bbv_{1}$ the expected value of the frequency representation of $\bbx$ is $\mbE[\tbx]=\bbV^{\Hr} \bbmu = \mu \bbe_{1}$ where $\bbe_{1}=[1,0,\ldots,0]^{\Tr}$ is the first element of the canonical basis of $\reals^N$. 

Property (ii) implies that the graph frequency representation of $\bbx$ is comprised of uncorrelated components. Indeed, since we know that the spectral representation covariance is $\bbC_{\tdx} = \bbV^{\Hr}\bbC_x\bbV$ and the covariance accepts the form in \eqref{eqn_joint_diagonalization} we conclude that 
\begin{equation} \label{eqn_psd}
   \bbC_{\tdx} \ :=\ \mbE\Big[(\tbx-\E{\tbx})(\tbx-\E{\tbx})^{\Hr}\Big]		
		       \  =\ \diag(\bbp).
\end{equation}
The covariance in \eqref{eqn_psd} not only implies that frequency coefficients are uncorrelated but also states that the components of the power spectral density $\bbp=[p_{1},\ldots,p_{N}]^{\Tr}$ represent the energy of the frequency coefficients of $\bbx$. Namely, that $\mbE[(\tdx_n-\mbE[\tdx_{n}])^2] = p_n$. This justifies referring to $\bbp$ as the PSD of $\bbx$. Property (ii) can be shown to be equivalent to imposing, on the covariance matrix $\bbC_{x}$, a form of invariance with respect to applications of the graph shift operator \cite{marques17}.

We remark that Def.~\ref{def_stationarity} boils down to the traditional WSS conditions on time signals, when modeling them as graph signals defined over the directed cycle graph (see Fig.~\ref{fig_dc}). It also reflects the fact that the ensemble mean of a WSS graph signal, like the GMRF described in Fig.~\ref{fig_gmrf_1000}, need not be a constant vector. More examples of WSS graph signals can be found in \cite{girault15, perraudin17, marques17}.

The goal of this paper is to study the relationship between the {\it graph shift average} $\hbmu_{N}$ in \eqref{eqn_graph_shift} and the {\it ensemble mean} $\bbmu$, which in the case of graph WSS signals satisfies \eqref{eqn_mean_definition}. We will see that for graph WSS, these two averages are close in terms of expressions that are similar to what follows from laws of large numbers of ergodic processes.

\begin{remark}\label{rmk_total_variation}\normalfont
The selection of the eigenvector $\bbv_1$ to codify the direction of the mean of a stationary graph signal is justified by the unique properties of the Perron-Frobenius root. We know that $\lambda_1$ is real and that $\bbv_1$ can be chosen to be real with nonnegative components \cite{cvetkovic79}. This means $\bbv_1$ has no zero crossings and in that sense it makes for a reasonable choice of DC component \cite{shuman13}. Another important property of $\bbv_1$ is that it corresponds to the eigenvalue with minimal variation \cite{sandryhaila14}. Specifically, define the total variation of a graph signal $\bbx$ as
\begin{equation}\label{eqn_tvdef}
	TV(\bbx) : = \left\| \bbx -  (1/\lambda_1)\bbS \bbx \right\|_{1}.
\end{equation}
The expression in \eqref{eqn_tvdef} compares components of $\bbx$ with the neighborhood average $\bbA \bbx$ normalized by $\lambda_1$ and it is therefore a proper measure of signal variability. The normalization by $\lambda_1$ is so that shifting the signal around does not cause an unbounded grow of the signal energy \cite{sandryhaila14} -- recall that $\lambda_1$ is the largest eigenvalue of $\bbS$. The eigenvector $\bbv_1$ is the one with minimal variation, i.e., $TV(\bbv_1)< TV(\bbv_n)$ for all $n\neq1$.
\end{remark}

\begin{remark} \normalfont
Def.~\ref{def_stationarity} establishes the conditions that a random graph signal $\bbx$ has to satisfy with respect to the graph shift operator $\bbS$ that describes a given graph $\ccalG$ to be considered WSS. However, there are many different shift operators $\bbS$ that can be chosen for any given graph $\ccalG$. While the choice of $\bbS$ is task dependent, a matrix $\bbS$ for which a set of signals results in WSS graph signals on $\bbS$ can be obtained by the methods described in \cite{segarra17-topoid}.
\end{remark}


%
\section{Graph Ergodicity} \label{sec_ergodicity}

If we restrict attention to the cyclic graph adjacency $\bbA_{\dc}$, classical results for stationary signals state that the graph shift average in \eqref{eqn_graph_shift} is unbiased, $\mbE[\hbmu_{N}] = \bbmu = \mu \bbone$. We further know that the the weak Law of Large Numbers (WLLN) holds \cite{billingsley95} and that we therefore have
\begin{equation} \label{eqn_time_wlln}
	\mbP \Big( | \hbmu_{N} - \bbmu|_n \geq \eps \Big)
		\leq \ccalO \left(\frac{1}{N\eps^{2}} \right). 
\end{equation}
The above implies that the probability of the ensemble mean and the graph shift average being arbitrarily far from each other decreases proportional to $1/N$ as the size $N$ of the underlying graph support increases and it further shows that to reduce the difference between $\hbmu_{N}$ and $\bbmu$ we need a quadratic increase in the size of the network -- i.e., to halve $\eps$ we need to increase $N$ by a factor of four if we want to maintain the same probability bound. We point out that writing $f(N,\eps) \leq \ccalO (1/\eps^{2}N)$ is not different from writing $f(N,\eps) = \ccalO (1/\eps^{2}N)$ as both would imply convergence at a rate that is at least proportional to $1/\eps^{2}N$ as $N$ grows. Convergence may be faster in either case. We choose to utilize the inequality to emphasize that a faster convergence rate may hold as we will observe that this indeed happens in some of our numerical analyses in Section \ref{sec_sims}. We will adopt the same convention for inequalities of the form $f(N) \leq o (1/N)$ which means that $f(N)$ goes to zero faster than $1/N$. This is the same as stating that $f(N) = o (1/N)$ but emphasizes that the bound is not tight in that even faster convergence may be observed.

Our goal is to generalize \eqref{eqn_time_wlln} to graph stationary signals. We begin by showing that $\hbmu_{N}$ can be made unbiased by properly defining $\alpha(\bbS)$ (Section \ref{subsec_diffusion}) and follow by considering generalizations of \eqref{eqn_time_wlln} (Section \ref{subsec_wlln}).

%
\subsection{Moments of Graph Shift Averages} \label{subsec_diffusion} 

As we show in the following proposition, the expectation of the graph shift average $\hbmu_{N}$ is aligned with $\bbv_1$ which is itself aligned with the mean $\bbmu$. To make the estimator unbiased we need to select $\alpha(\bbS)$ properly as we also show in the following proposition.

%
\begin{proposition} \label{prop_unbiased}
Let $\bbx$ be a WSS graph signal on the $N$-node graph $\ccalG$ described by a normal graph shift operator $\bbS=\bbV \bbLambda \bbV^{\Hr}$. Let $\bbmu$ be the ensemble mean [cf. \eqref{eqn_mean_definition}] and $\hbmu_{N}$ be the graph shift average [cf. \eqref{eqn_graph_shift}]. Then,
\begin{equation}\label{eqn_unbiased}
   \mbE [ \hbmu_{N} ] = \frac{\sum_{\ell=0}^{N-1}\lambda_{1}^{\ell}}
                             {\alpha(\bbS)} 
                        \bbmu.
\end{equation}
In particular, making $\alpha(\bbS) = \sum_{\ell=0}^{N-1} \lambda_{1}^{\ell}$ the graph shift average is an unbiased estimator as \eqref{eqn_unbiased} reduces to $\mbE [ \hbmu_{N} ] = \bbmu$.
\end{proposition}
%
%
\begin{proof}
See appendix.
\end{proof}

%
Proposition \ref{prop_unbiased} reveals the impact that the spectral properties of the underlying graph support have on the graph shift average, and thus the necessary correction to make this average an unbiased estimator. Henceforth, we use $\alpha(\bbS) = \sum_{\ell=0}^{N-1} \lambda_{1}^{\ell}$ in all subsequent analyses so that $\hbmu_{N}$ is an unbiased estimator of $\bbmu$,
\begin{equation} \label{eqn_graph_shift_unbiased}
	\hbmu_{N} = \frac{1}{\sum_{\ell=0}^{N-1} \lambda_{1}^{\ell}} 
	            \sum_{\ell=0}^{N-1} \bbS^{\ell}\bbx .
\end{equation}
The term ${1}/{\sum_{\ell=0}^{N-1} \lambda_{1}^{\ell}}$ is a normalization that plays a role that is equivalent to the $1/N$ factor in the time average in \eqref{eqn_time_average}. In fact, if we make $\lambda_1=1$ we have ${1}/{\sum_{\ell=0}^{N-1} \lambda_{1}^{\ell}} = 1/N$. This is always possible as it just entails a normalization of the unit that is used to measure proximity between the nodes of the graph. We keep working with arbitrary $\lambda_1$ for generality.

Writing the componentwise version of $\mbE [ \hbmu_{N} ] = \bbmu = \mbE[\bbx]$ we observe that $\mbE [ [\hbmu_{N}]_{k} ] = \mu_{k} = \mbE [x_{k}]$ for $k=1,\ldots,N$. which means that node $k$ can use the graph shift average to estimate the mean of the local random signal $x_k$. This is of practical importance, we recall, because the diffusion process that generates \eqref{eqn_graph_shift_unbiased} can be implemented in a distributed manner. This means that node $k$ can rely on communication with neighboring nodes to aggregate information from the network to refine its estimate of the mean of the signal component it has observed.

Of course, the mean of $x_{k}$ is itself $\mbE[x_{k}]=\mu_{k}$ so the question arises if the graph shift average is actually a refinement of simply using $x_{k}$ as an estimate of $\mu_{k}$. To investigate this question, notice that $\hbmu_{N}$ is a random signal that is WSS with respect to the shift operator $\bbS$. The following proposition states this formally and computes the PSD of $\hbmu_{N}$.

%
\begin{proposition} \label{prop_psd_graph_shift_average}
The graph shift average in \eqref{eqn_graph_shift_unbiased} is WSS with respect to the graph shift $\bbS$ [cf. Def.~\ref{def_stationarity}]. The covariance can be written as $\bbC_{\hat{\mu}} = \bbV \diag(\bbq)\bbV^{\Hr}$ with the elements of the PSD $\bbq = [q_1, \ldots, q_n]^{\Tr}$ explicitly given by
\begin{equation}\label{eqn_qk}
   q_{n} \ = \ p_{n} \ \bigg|\sum_{\ell=0}^{N-1} \lambda_{n}^{\ell}\bigg|^{2} 
                     \ \Big/ \
                     \bigg|\sum_{\ell=0}^{N-1} \lambda_{1}^{\ell}\bigg|^{2} . 
\end{equation}

\end{proposition}
%
%
\begin{proof}
See appendix.
\end{proof}

%
Note that the first element of the PSD is $q_{1}=p_{1}$, so that this frequency coefficient remains unchanged while the rest are attenuated. In other words, the graph shift average \eqref{eqn_graph_shift_unbiased} acts as a low-pass graph filter that decreases the power of the signal on all frequencies except the lowest frequency coefficient corresponding to the mean. It is further observed that estimating the mean by simply using $x_{k}$ incurs in a variance $\textrm{var}(x_{k}) = \mbE[(x_{k}-\mu_{k})^{2}]$ that is given by
\begin{equation} \label{eqn_var_xk}
\mbE[(x_{k}-\mu_{k})^{2}] 
	= [\bbC_{x}]_{k,k} 
	= \sum_{n=1}^{N} p_{n} |\bbv_{k,n}|^{2}
\end{equation}
which depends on the power of all frequencies $p_{n}$. It becomes clear then, that if we lower the power $p_{n}$, the variance is reduced. Thus, the fact that $\hbmu_{N}$ acts as a low-pass graph filters plays a key role in the convergence of the graph shift average to the ensemble mean. This is further elaborated in the next section.

%
\begin{remark}\normalfont\label{rmk_laplacian}
Our assumption on the graph shift being nonnegative precludes use of the Laplacian $\bbL=\bbD-\bbA$ as a graph shift. The nonnegative assumption can be relaxed to the existence of a real maximal eigenvalue whose associated eigenvector defines the direction of the mean. This would allow the use of symmetric Laplacians since the eigenvalues are real and nonnegative. We could then define a graph shift average as in \eqref{eqn_graph_shift_unbiased} to which Propositions \ref{prop_unbiased} and \ref{prop_psd_graph_shift_average} would apply. While mathematically correct, the results are irrelevant because when $\bbS=\bbL$ the mean of the process is naturally aligned with the eigenvector associated with the smallest eigenvalue -- as opposed to the one associated with the largest eigenvalue
. Indeed, since $\bbL\bbone=\bbzero$ it follows that requiring $\E{\bbx}=\mu\bbone$ is a natural choice for the mean of the WSS graph signal $\bbx$. The eigenvalue associated with eigenvector $\bbone$ is $0$, which is the smallest a Laplacian eigenvalue can be. This observation is consistent with the fact that the signal $\bbL\bbx$ computes the difference between the value at a node and the average values of its neighbors. The Laplacian diffusion then acts as a high-pass operation and it is therefore natural that it amplifies high frequency components -- as opposed to adjacency diffusions which are averaging operations and therefore amplify low frequency components.
\end{remark}

%
\subsection{Weak Law of Large Numbers} \label{subsec_wlln}

Convergence of graph shift average \eqref{eqn_graph_shift_unbiased} to the ensemble mean $\bbmu$ of the process is determined by the following theorem.

%
\begin{theorem}[Weak law of large numbers]
	\label{thm_glln}
Let $\bbx$ be a WSS graph signal on a $N$-node graph $\ccalG$ characterized by a normal graph shift operator $\bbS=\bbV \bbLambda \bbV^{\Hr}$. Let $\lambda_{1} \in \reals$ be the largest positive eigenvalue such that $|\lambda_{n}| \le \lambda_{1}$, $\lambda_{1} \ne \lambda_{n}$ for all $n=2,\ldots,N$. If $\lambda_{1} >1$ and $|\lambda_{n}|/\lambda_{1}=o(N^{-\delta/2(N-1)})$ for some $\delta>0$ for all $n\geq 2$, or if $\lambda_{1}=1$, then, if $p_{n}<\infty$ for all $n$, we have that
\begin{equation} \label{eqn_glln}
\min_{k=1,\ldots,N} 
	\mbP \left( \left|[\hbmu_{N}-\bbmu]_{k} \right| > \eps \right)
		\leq 
	\frac{p_{1}}{N\eps^{2}} + o (N^{-\delta}).
\end{equation}
Additionally,
\begin{equation} \label{eqn_max_thm}
\max_{k=1,\ldots,N} 
	\mbP \left( \left|[\hbmu_{N}-\bbmu]_{k} \right| > \eps \right)
		\leq 
	\frac{p_{1}}{\eps^{2}} + o (N^{-\delta}).
\end{equation}
\end{theorem}

%
Thm.~\ref{thm_glln} states that the probability of the graph shift average \eqref{eqn_graph_shift_unbiased} being far from the true ensemble mean $\bbmu$ at any node decreases polynomially $\ccalO(1/N)$ as the size of the graph $N$ increases with a small term $o(N^{-\delta})$ that vanishes, provided that the graph spectra satisfies certain conditions. In particular, eq. \eqref{eqn_glln} states that there exists a node for which the probability of the graph shift average being arbitrarily far from the ensemble mean decreases to zero as the number of nodes in the graph increases. Alternatively, we can make $\eps=1/\sqrt{N}$ and get a concentration inequality by which the graph shift average gets arbitrarily close to the mean at a rate of $1/\sqrt{N}$ with constant probability. It is shown later that these conditions on the graph spectra do hold for some practical graphs such as the directed cycle (Cor.~\ref{coro_wlln}) and Erd{\H{o}}s-R{\'e}nyi graphs (Cor.~\ref{coro_er}). Eq. \eqref{eqn_glln} is reminiscent of the traditional WLLN for time sequences [cf. \eqref{eqn_time_wlln}] and thus prove ergodicity of the first moment in WSS graph signals.

In order to prove Thm.~\ref{thm_glln} some preliminary results are needed. First, we compute a bound on the probability of error when estimating the mean at a single node, resulting in Lemma~\ref{l_unbiased_error_node}. We observe that such bound depends on the PSD of the estimator $\bbq$ \eqref{eqn_qk} and on the value of the eigenvectors at the given node. We study the behavior of the PSD $\bbq$ with the number of nodes $N$ and find, in Lemma~\ref{l_behavior_qk}, the conditions on the graph spectra under which the PSD $\bbq$ decreases with increasing $N$. Finally, we couple this result together with Lemma~\ref{l_unbiased_error_node} to prove ergodicity in Thm.~\ref{thm_glln}. Corollaries~\ref{coro_wlln}~and~\ref{coro_er} show applications of Thm.~\ref{thm_glln} to the particular cases of directed cycle and Erd\H{o}s-R\'{e}nyi graphs.

%
\begin{lemma}[Error bound]
	\label{l_unbiased_error_node}
Let $\bbx$ be a WSS graph signal on a $N$-node graph $\ccalG$ described by a normal graph shift operator $\bbS=\bbV \bbLambda \bbV^{\Hr}$. Let $\mbE[\bbx]=\mu \bbv_{1}$ be the mean of the process and $\bbC_{x}=\bbV \diag(\bbp) \bbV^{\Hr}$ be the covariance matrix, where $p_{n}<\infty$ for all $n=1,\ldots,N$. Then, at node $k \in \{1,\ldots,N\}$ we have
\begin{equation}
	\mbP \left( \left| [ \hbmu_{N} - \bbmu]_{k} \right| > \eps \right) 
		\le \frac{1}{\eps^{2}} \sum_{n=1}^{N} q_{n} |v_{k,n}|^{2}.
		\label{eqn_unbiased_error_node}
\end{equation}
\end{lemma}
%
%
\begin{proof}
See appendix.
\end{proof}
%

%
The bound given in Lemma~\ref{l_unbiased_error_node} is a Chebyshev type bound for estimating the mean at a single node. The Chebyshev bound is the basic building block to prove the traditional WLLN. Note that, if we let $\bbx$ have i.i.d. random variables with variance $\sigma^{2}$ as elements, then $p_{n}=0$ for all $n=2,\ldots,N$ so that $\bbq=[q_{1},0,\ldots,0]^{\Tr}$ where $q_{1}=p_{1}=\sigma^{2}$. Then, only the value $|v_{k,1}|^{2}$ remains in \eqref{eqn_unbiased_error_node} and, in the case of a directed cycle, it has a value of $|v_{k,1}|^{2}=1/N$, see the specifics in Corollary~\ref{coro_wlln}. Also, observe that $\{v_{k,n}, n=1,\ldots,N\}$ are the values contained in the $k$th row of $\bbV$, as noted in the proof of the lemma. Because $\bbV$ is an unitary matrix, rows also form an orthonormal set. Understanding \eqref{eqn_unbiased_error_node} from a GFT viewpoint, we note that the performance of the estimation at a single node $k$ depends on the variance of the nodes within the $(N-1)$-hop neighborhood and on all frequency coefficients (all eigenvalues) of the graph shift operator. Additionally, it is affected by the value of each frequency component (each eigenvector) on node $k$ alone.

Another important observation stemming from Lemma~\ref{l_unbiased_error_node} is that the graph shift average \eqref{eqn_graph_shift_unbiased} improves the estimation over simply looking at $x_{k}$ for \emph{every} node $k=1,\ldots,N$. More precisely, the variance of $x_{k}$ given in \eqref{eqn_var_xk} has a similar form to the variance of $[\hbmu_{N}]_{k}$ [cf. \eqref{eqn_var_hbmu}]. But, since $q_{n} \leq  p_{n}$ for every $n$ [cf. \eqref{eqn_qk}], then we have that, for \emph{every} $k=1,\ldots,N$
\begin{equation}
\textrm{var} \left( [\hbmu_{N}]_{k} \right) \leq \textrm{var} \left( x_{k} \right)
\end{equation}
proving that the graph shift average improves the estimation of the ensemble mean at every node.

In order to prove convergence of the graph shift average we need to analyze how $q_{n}$ and $|v_{k,n}|^{2}$ behave relative to the size of the graph $N$.

%
\begin{lemma}[Behavior of $q_{n}$ with size of graph $N$]
	\label{l_behavior_qk}
Let $\ccalG=(\ccalV,\ccalE,\ccalW)$ be a $N$-node weighted graph that admits a normal graph shift operator $\bbS=\bbV \bbLambda \bbV^{\Hr}$. Let $\lambda_{1} \in \reals$ be the largest positive eigenvalue such that $|\lambda_{n}| \le \lambda_{1}$, $\lambda_{1} \ne \lambda_{n}$ for all $n=2,\ldots,N$. If $\lambda_{1} >1$ and $|\lambda_{n}|/\lambda_{1}=o(N^{-\delta/2(N-1)})$ for some $\delta>0$, or if $\lambda_{1}=1$, then
\begin{equation} \label{eqn_behavior_qk}
	q_{n}=o(N^{-\delta}) 
		\quad , \: n=2,\ldots,N.
\end{equation}
For $n=1$ we always have $q_{1}=p_{1}$.
\end{lemma}
%
%
\begin{proof}
See appendix.
\end{proof}

%
Lemma~\ref{l_behavior_qk} shows that $q_{n}$ is polynomially decreasing for $n=2,\ldots,N$, under specific restrictions on the spectrum of the graph. For cases where $\lambda_{1}>1$, we observe that the eigenvalues can be of the same order of $\lambda_{1}$, but they still have to grow at a slower rate than $\lambda_{1}$. Lemma~\ref{l_behavior_qk} is the building block to prove convergence of the graph shift average $\hbmu_{N}$. We are finally equipped with all the necessary tools to prove Thm.~\ref{thm_glln}.

%
\begin{proof}[Proof of Thm.~\ref{thm_glln}]
Consider \eqref{eqn_unbiased_error_node}. First, observe that because $\|\bbv_{1}\|^{2}=1$ there is always a node $k$ for which $|v_{k,1}| \le 1/\sqrt{N}$ so that
\begin{equation} \label{eqn_min_vk}
\min_{k=1,\ldots,N} |v_{k,1}|^{2} 
	\le \frac{1}{N}.
\end{equation}
Similarly, we have that
\begin{equation} \label{eqn_max_vk}
\max_{k=1,\ldots,N} |v_{k,1}|^{2} 
	\le 1.
\end{equation}
Let $q_{\max}=\max_{n=2,\ldots,N} \{q_{n}\}$. Then, from \eqref{eqn_unbiased_error_node} together with \eqref{eqn_min_vk} we have that
\begin{equation} \label{eqn_full_bound_min}
\min_{k=1,\ldots,N} 
	\mbP \left( \left| [\hbmu_{N}-\bbmu]_{k} \right| > \eps \right) 
		\le
	\frac{1}{\eps^{2}} \left( \frac{p_{1}}{N} + q_{\max} \right)
\end{equation}
where the fact that $\sum_{n=2}^{N} |v_{k,n}|^{2} \le 1$ for all $k$ was used (because the rows of $\bbV$ also form an orthonormal basis since $\bbV$ is unitary). Analogously, using \eqref{eqn_unbiased_error_node} in combination with \eqref{eqn_max_vk} yields
\begin{equation} \label{eqn_full_bound_max}
\max_{k=1,\ldots,N} 
	\mbP \left( \left| [\hbmu_{N}-\bbmu]_{k} \right| > \eps \right) 
		\le
	\frac{1}{\eps^{2}} \left( p_{1} + q_{\max} \right)
\end{equation}
Now because the eigenvalues of $\bbS$ satisfy the assumptions of Lemma \ref{l_behavior_qk} by hypothesis and because $q_{\max}=q_{n}$ for some $n=2,\ldots,N$, then we know that $q_{\max}=o(N^{-\delta})$ thus turning \eqref{eqn_full_bound_min} into \eqref{eqn_glln} and \eqref{eqn_full_bound_max} into \eqref{eqn_max_thm}, completing the proof.
\end{proof}

To close this section we consider two examples of widespread use that satisfy the conditions in Thm.~\ref{thm_glln}, namely the directed cycle and Erd\H{o}-R\'{e}nyi graphs.

\begin{myparagraph}{\bf Directed cycle (Classical WLLN)}
The directed cycle $\ccalG_{\dc}$ represents the graph support for time-stationary signals, see Fig.~\ref{fig_dc}. Then, by applying Thm.~\ref{thm_glln} we expect to recover the traditional WLLN.

%
\begin{corollary}[Convergence of Directed Cycle]\label{coro_wlln}
Let $\ccalG_{\dc}$ be the directed cycle graph. Then, for any node $k \in \{1,\ldots,N\}$ the error bound is
\begin{equation} \label{eqn_wlln}
\mbP \left( \left| \frac{1}{N} \sum_{n=1}^{N}x_{n} -\mu \right| > \eps \right) 
	\le \frac{p_{1}}{N\eps^{2}}.
\end{equation}
\end{corollary}
%
%
\begin{proof}
See appendix.
\end{proof}

%
\noindent Corollary \ref{coro_wlln} is a statement of the WLLN for signals that are stationary in time. The result in this case is stronger than the one in \eqref{eqn_glln} because it lacks the order term $o(N^{-\delta})$. This term vanishes because in the case of a cycle graph the $n$th component of the estimator's PSD is $q_{n}=0$, $n=2,\ldots,N$. It is also stronger in that the minimum disappeared since, after $N$ shifts, all nodes have aggregated the sample mean, so any node yields the same estimator and thus the same probability of error. Finally, note that if $\bbx$ are i.i.d. r.v. then the DC component of the signal is $p_{1}=\sigma^{2}$ and \eqref{eqn_wlln} is the Chebyshev's bound that leads to the classical WLLN.
\end{myparagraph}

\begin{myparagraph}{\bf Erd\H{o}s-R{\'e}nyi (ER) graphs}
Another family of graphs that satisfies the conditions of Thm.~\ref{thm_glln} are ER graphs. These graphs have the particularity that the largest eigenvalue grows linearly with the number of nodes whereas the rest of the eigenvalues have a growth rate that does not exceed $\sqrt{N}$. This means that the graph is well-suited for estimation since the PSD of the graph shift average concentrates around the largest eigenvalue corresponding to the mean of the process. This is shown in the proof of the following corollary.

%
\begin{corollary}[WLLN for Erd\H{o}s-R{\'e}nyi graphs] \label{coro_er}
Let $\ccalG_{\ER}$ be an ER graph of size $N$ with edge probability $p_{\ER}$ such that $Np_{\ER} \to \beta \ge 1$. Then, for any node $k \in \{1,\ldots,N\}$ and any $0 < \delta < N-1$ we have that
\begin{equation} \label{eqn_wglln_er}
	\mbP \left( \left| \left[ \hbmu_{N}-\bbmu \right]_{k} \right| > \eps \right) \le \frac{p_{1}}{N\eps^{2}} + o(N^{-\delta}).
\end{equation}
\end{corollary}
%
%
\begin{proof}
See appendix.
\end{proof}
%

%
\noindent Corollary \ref{coro_er} states that the estimator obtained at any node is arbitrarily close to the ensemble mean at that node, with a convergence rate that is polynomial on the size of the graph. While ER graphs are perhaps of limited modeling power, they help to illustrate a family of graphs that satisfies the conditions of Theorem~\ref{thm_glln}. More general graphs are addressed in the next section.
\end{myparagraph}

\begin{remark}[Infinite diffusions on a fixed size graph] \normalfont
	\label{rmk_infinite_diffusions}
Alternatively, one could think of considering a graph of fixed size $N$ and unbiased estimator $\hbmu_{L} = (\sum_{\ell=0}^{L-1} \lambda_{1}^{\ell})^{-1} \sum_{\ell=0}^{L-1} \bbS^{\ell} \bbx$ in which the signal is diffused $(L-1)$ times, $L$ independent of $N$. We observe that, if $\lambda_{1} > 1$ and $|\lambda_{n}|<\lambda_{1}$, or if $\lambda_{1}=1$, then $q_{n} \to 0$ as $L \to \infty$ for $n = 2,\ldots,N$, and $q_{1}=p_{1}$. This implies that bound in Lemma~\ref{l_unbiased_error_node} hits a fundamental limit given by $\lim_{L \to \infty} \mbP ( |[\hbmu_{L} - \hbmu|_{k}| > \eps) \leq q_{1}|v_{k,1}|^{2}/\eps^{2}$. Since $|v_{k,1}|^{2}$ depends on the graph topology, and is fixed for fixed $N$, then the bound cannot be shown to decrease any further, even if the signal is diffused an infinite number of times. The latter situation is typically known as the \emph{consensus problem} and we note that the result obtained in Thm.~\ref{thm_glln} is fundamentally different since it deals with convergence of the graph shift average \eqref{eqn_graph_shift_unbiased} as the size of the graph gets larger.
\end{remark}

\begin{remark}[Power method] \normalfont
	\label{rmk_power_method}
The graph shift average in \eqref{eqn_graph_shift_unbiased} has a vague similarity with the power method, which is used to compute the eigenvector associated to the largest eigenvalue of a matrix \cite[Sec.~10.3]{larson17}. This vague similarity notwithstanding, we note that the objective of the graph shift average is to estimate the mean of a WSS graph process on a given graph which is fundamentally different from the problem of estimating the associated eigenvector. It is true that the power method can be used to recover the eigenvalue associated with the largest eigenvector, but this would be a different way of estimating the mean of the graph stationary process. Among many other differences the convergence rate of the power method is the eigenvalue ratio $\lambda_{2}/\lambda_{1}=o(N^{-\delta/2(N-1)})$ whereas the convergence rate for the graph shift average is the much faster rate $o(N^{-\delta})$. Analogous comments apply to the inverse iteration method which uses the graph shift operator inverse to compute the eigenvector $\bbv_{1}$ iteratively.
\end{remark}


\section{Optimal Mean Estimation with Graph Filters} \label{sec_optimal}

In cases where the graph spectrum does not fall under the conditions of Lemma~\ref{l_behavior_qk} and thus Thm.~\ref{thm_glln} cannot be applied, we analyze the convergence of the graph shift average \eqref{eqn_graph_shift_unbiased} in the following lemma.

%
\begin{lemma}[Non-convergent graphs]
	\label{l_nonconvergent}
Let $\ccalG=(\ccalV,\ccalE,\ccalW)$ be a weighted $N$-node graph that admits a normal graph shift operator $\bbS= \bbV \bbLambda \bbV^{\Hr}$. Let $\lambda_{1} \in \reals$ be the largest positive eigenvalue such that $|\lambda_{n}| \le \lambda_{1}$, $\lambda_{n} \ne \lambda_{1}$ for all $n=2,\ldots,N$. Let $\ccalM$ be the set of indices $m$ such that $|\lambda_{m}|/\lambda_{1}$ does not satisfy $o(N^{-\delta/2(N-1)})$ for any $\delta>0$. If $\lambda_{1}>1$ and $\ccalM$ is nonempty, or if $\lambda_{1}<1$, then for any node $k \in \{1,\ldots,N\}$ it holds that
\begin{align}
\mbP \left( \left| [\hbmu_{N}-\bbmu]_{k} \right| > \epsilon \right) 
	& \le \frac{p_{1}}{\epsilon^{2}} |v_{k,1}|^{2}+o(1)
		\label{eqn_nonconvergent} \\
	& + \sum_{m \in \ccalM} p_{m} \left| \frac{1-\lambda_{1}}{1-\lambda_{m}} \right|^{2} |v_{k,m}|^{2} (1+o(1)).
		\nonumber
\end{align}
If $\lambda_{1}<1$, then $\ccalM=\{2,\ldots,N\}$.
\end{lemma}
%
\begin{proof}
See appendix.
\end{proof}

%
For the case in which $\lambda_{1}<1$ we reach a fundamental limit under which is not possible to achieve a better estimation. This situation occurs because on each successive step of the diffusion process, the information harnessed from neighboring nodes is less and less (because all the eigenvalues are less than $1$), eventually making it impossible to accurately estimate the mean. Alternatively, when $\lambda_{1}>1$, if $|\ccalM|=o(N)$ then the graph shift average \eqref{eqn_graph_shift_unbiased} is still consistent since at most finitely many values of $q_{m}$ do not follow the rate $o(N^{-\delta})$. When $|\ccalM|=\ccalO(N)$, then no assertions about the convergence rate of the graph shift average can be done, and there is a fundamental constant approximation. In what follows, we address this issue, extending convergence for more general graphs.

Convergence rates can be tuned by the use of graph filters. Linear shift-invariant (LSI) graph filters are linear transformations that can be applied to the signal in a local fashion, operating only with the values of the signal at the neighborhood of each node \cite{segarra17-linear}. More precisely, let $\{h_{\ell}\}_{\ell=0}^{L-1}$ be a set of $L$ filter taps, then a LSI graph filter is the $N \times N$ matrix $\bbH=\sum_{\ell=0}^{L-1} h_{\ell} \bbS^{\ell}$, where $\bbS^{0}=\bbI_{N}$. Note that the output signal
\begin{equation}
	\bby = \bbH \bbx = \sum_{\ell=0}^{L-1} h_{\ell} (\bbS^{\ell}\bbx)
		\label{eqn_lsigf}
\end{equation}
can be computed by accessing the values on the nodes in the $(L-1)$-hop neighborhood at most. For each hop $\ell$, the resulting shifted value on the node is further weighted by filter tap $h_{\ell}$, $\ell=0,\ldots,L-1$. The effect of LSI graph filters on the signal can also be analyzed by projecting the output $\bby$ on the frequency basis, $\tby = \bbV^{\Hr} \bby$. In order to do this, first define the GFT of the graph filter $\tbh \in \mbC^{N}$ as follows \cite{sandryhaila14}
\begin{equation}
	\tbh = \bbPsi \bbh
		= \begin{bmatrix}
			1      & \lambda_{1} & \lambda_{1}^{2} & \cdots & \lambda_{1}^{L-1} \\ 
			\vdots & \vdots      & \vdots          & \ddots & \vdots \\
			1      & \lambda_{N} & \lambda_{N}^{2} & \cdots & \lambda_{N}^{L-1}
		\end{bmatrix} \cdot \begin{bmatrix}
			h_{0} \\ h_{1} \\ h_{2} \\ \vdots \\ h_{L-1}
		\end{bmatrix}.
		\label{eqn_GFT_filter}
\end{equation}
Matrix $\bbPsi \in \mbC^{N \times L}$ is a Vandermonde matrix that acts as the linear transformation that computes the $N$ frequency coefficients $\tbh \in \mbC^{N}$ of the graph filter from the filter taps given in $\bbh \in \reals^{L}$. Observe that, unlike temporal signals and filters, the GFT of graph signals and graph filters are computed differently. The former depending on the eigenvectors of the graph shift operator, whereas the latter depends only on the eigenvalues. By denoting $\circ$ as the elementwise (Hadamard) product of two vectors, we can then obtain the frequency coefficients of the output signal directly from the frequency coefficients of the filter and the signal as follows
\begin{equation}
	\tby = \diag(\tbh) \tbx 
	     = \tbh \circ \tbx.
	     	\label{eqn_convtheorem}
\end{equation}
Note that \eqref{eqn_convtheorem} is analogous to the convolution theorem for temporal signals \cite[Sec.~2.9.6]{oppenheimschafer10}. Graph filters are useful in shaping graph signals and their frequency coefficients by means of local linear operations only.

As shown in Lemma~\ref{l_nonconvergent} the effect of the graph frequencies $q_{n}$ on $\hbmu_{N}$ determine its convergence. Therefore, by carefully designing filter taps $\{h_{\ell}\}_{\ell=0}^{L-1}$ we can obtain a desired graph signal with specific frequency characteristics that can aid in the convergence of the estimator. More precisely, we propose an optimal design of a LSI graph filter that can be applied to the single realization of the graph signal that, not only improves the performance (by minimizing the mean squared error and the volume of the ellipsoid error) but that is also shown to converge for any graph, see Theorems~\ref{thm_optimal_MSE} and \ref{thm_optimal_logdet}. First, we need to restrict the possible filter taps to yield an unbiased estimator.

%
\begin{proposition}
	\label{prop_optimal_unbiased}
Let $\bbx$ be a WSS graph signal on a $N$-node graph $\ccalG$ described by a normal graph shift operator $\bbS=\bbV \bbLambda \bbV^{\Hr}$. Let $\mbE[\bbx]=\bbmu=\mu \bbv_{1}$ where $\bbv_{1}$ is the eigenvector associated to $\lambda_{1}$, the largest, positive eigenvalue such that $\lambda_{n} \ne \lambda_{1}$, $|\lambda_{n}| \le \lambda_{1}$ for all $n=2,\ldots,N$. Let $\{h_{\ell}\}_{\ell=0}^{N-1}$ be a set of $N$ tap filters, where at least one is nonzero. Let $\bby_{N}$ be the output of processing $\bbx$ through the graph filter with taps given by $\bbh=[h_{0},\ldots,h_{N-1}]$. That is, $\bby_{N}=\sum_{\ell=0}^{N-1} h_{\ell}\bbS^{\ell} \bbx$. Then, the estimator
\begin{equation}\label{eqn_optimal_unbiased}
\bbz_{N} 
	= \frac{1}{\sum_{\ell=0}^{N-1} h_{\ell} \lambda_{1}^{\ell}} \bby_{N} 
	= \frac{1}{\sum_{\ell=0}^{N-1} h_{\ell} \lambda_{1}^{\ell}} \sum_{\ell=0}^{N-1} h_{\ell}\bbS^{\ell} \bbx
\end{equation}
is unbiased for any choice of $\{h_{\ell}\}_{\ell=0}^{N-1}$.
\end{proposition}
%
%
\begin{proof}
See appendix.
\end{proof}
%

\begin{figure*}
    \captionsetup[subfigure]{justification=centering}
    \centering
    \begin{subfigure}{0.66\columnwidth}
        \centering
        \includegraphics[width=0.99\textwidth]{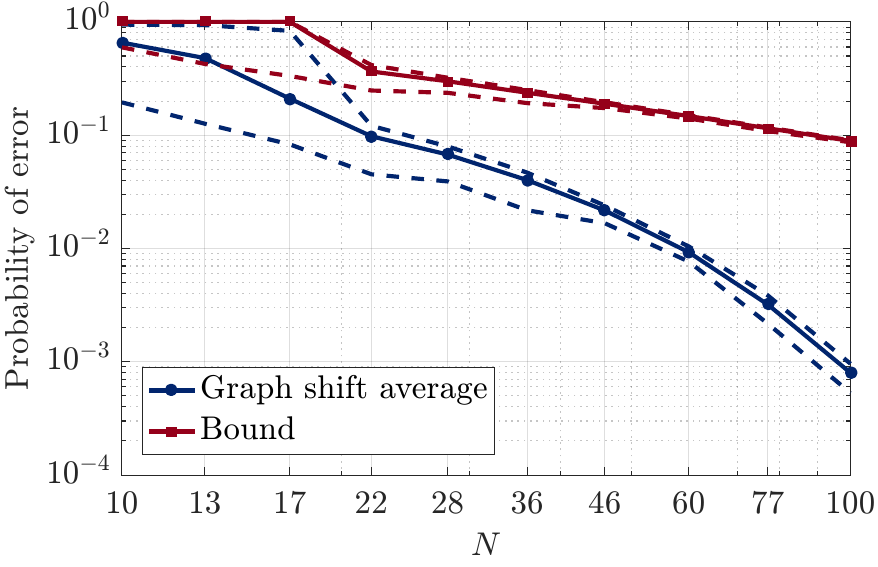}
        \caption{Graph shift average $\hbmu_{N}$ \eqref{eqn_graph_shift_unbiased}}
        \label{p_error_hbmu-ER}
    \end{subfigure}
    \hfill
    \begin{subfigure}{0.66\columnwidth}
        \centering
        \includegraphics[width=0.95\textwidth]{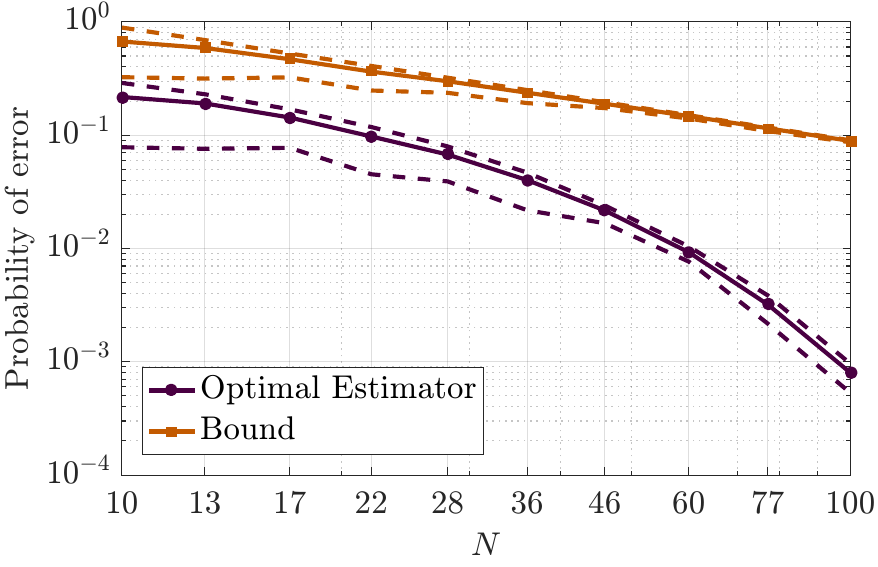}
        \caption{Optimal estimator $\bbz_{N}$ \eqref{eqn_optimal_unbiased}-\eqref{eqn_optimal_MSE}}
        \label{p_error_z-ER}
    \end{subfigure}
    \hfill
    \begin{subfigure}{0.66\columnwidth}
        \centering
        \includegraphics[width=0.95\textwidth]{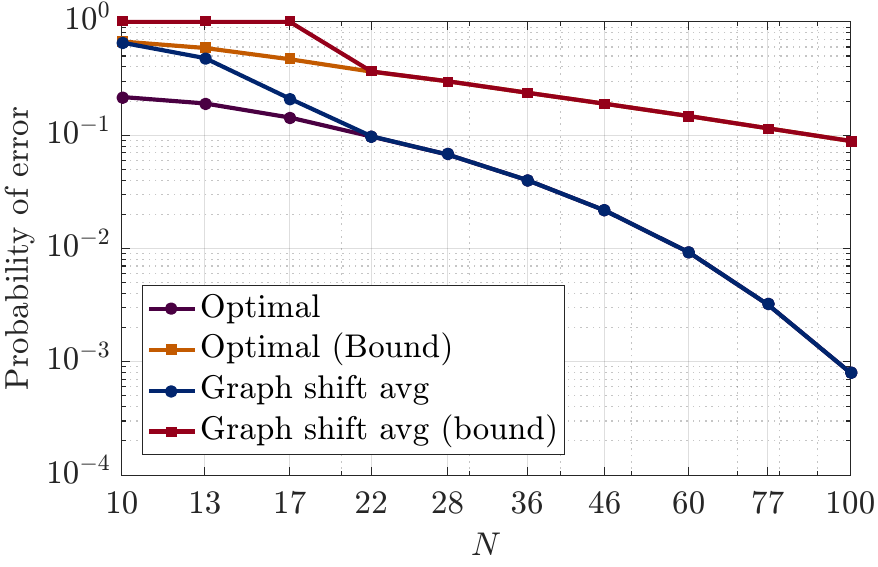}
        \caption{Comparison}
        \label{p_error_comp-ER}
    \end{subfigure}
    \caption{Erd{\H{o}}s-R{\'e}nyi Graph. The full lines correspond to the mean bound and mean probability of error for $50$ graph realizations. The dashed lines correspond to the maximum and minimum values obtained in some realization. \subref{p_error_hbmu-ER}-\subref{p_error_z-ER} Estimated probability of error and theoretical bound for ER graphs of varying size $N$ from $10$ to $100$ for the graph shift average $\hbmu_{N}$ and the optimal estimator $\bbz_{N}$, respectively. \subref{p_error_comp-ER} Comparison of the mean probability of error and mean theoretical bound for both the graph shift average and the optimal estimator.}
    \label{fig_ER}
\end{figure*}

Prop.~\ref{prop_optimal_unbiased} determines how to obtain an unbiased estimator from filtering a WSS graph signal with an arbitrary LSI graph filter and highlights the effect that the spectral properties of the underlying graph support have on the filtered signal. Filtered output $\bbz_{N}$ is itself a WSS graph signal with covariance matrix as specified in the following proposition.

%
\begin{proposition} \label{prop_psd_optimal}
The output $\bbz_{N}$ of LSI graph filter \eqref{eqn_optimal_unbiased} is WSS with respect to the graph shift $\bbS$ [cf. Def.~\ref{def_stationarity}]. The covariance can be written as $\bbC_{z} = \bbV \diag(\bbr)\bbV^{\Hr}$ with the elements of the PSD $\bbr = [r_{1}, \ldots, r_{n}]^{\Tr}$ explicitly given by
\begin{equation}\label{eqn_psd_optimal}
   r_{n} \ = \ p_{n} \ \bigg|\sum_{\ell=0}^{N-1} h_{\ell} \lambda_{n}^{\ell}\bigg|^{2} 
                     \ \Big/ \
                     \bigg|\sum_{\ell=0}^{N-1} h_{\ell} \lambda_{1}^{\ell}\bigg|^{2} 
           = p_{n} \frac{|\tdh_{n}|^{2}}{|\tdh_{1}|^{2}}.
\end{equation}
\end{proposition}
%
%
\begin{proof}
See appendix.
\end{proof}

First, it is noted that, if $h_{\ell}=1$ for all $\ell=0,\ldots,N-1$, then \eqref{eqn_psd_optimal} boils down to \eqref{eqn_qk}. Also, it is observed that $r_{1}=p_{1}$, while the rest of the frequencies are rescaled proportional to the amplitude of the specific frequency coefficient $|\tdh_{n}|^{2}$. It is finally noted that any LSI graph filter, when used as an unbiased estimator \eqref{eqn_optimal_unbiased}, leaves the first frequency coefficient unchanged, which is consistent with the notion of the mean being associated with the first frequency coefficient, and also with the intuition of trying to estimate this specific frequency coefficient (and hence, left unmodified).

Now that we have a systematic way to construct unbiased estimators from arbitrary LSI graph filters, we can determine optimality criteria to obtain the best performing estimator. For instance, consider the estimator that minimizes the MSE which is given by
\begin{equation}
\tr \left[ \bbC_{z} \right] = \sum_{n=1}^{N} p_{n} \frac{|\tdh_{n}|^{2}}{|\tdh_{1}|^{2}}.
		\label{eqn_MSE}
\end{equation}
The design that minimizes the MSE is given by $\tdh_{n}=0$ for $n=2,\ldots,N$ as shown in the following theorem.

%
\begin{theorem} \label{thm_optimal_MSE}
Let $\bbx$ be a WSS graph signal on a $N$-node graph $\ccalG$ described by a normal graph shift operator $\bbS = \bbV \bbLambda \bbV^{\Hr}$. Let $\lambda_{1}$ be the largest positive eigenvalue such that $\lambda_{n} \ne \lambda_{1}$, $|\lambda_{n}| \le \lambda_{1}$ for all $n=2,\ldots,N$. Let $\{h_{\ell}\}_{\ell=0}^{N-1}$ be a set of $N$ filter taps and $\{\tdh_{n}\}_{n=1}^{N}$ the frequency coefficients of the filter. Then, the filter that minimizes the MSE \eqref{eqn_MSE} is given by
\begin{equation} \label{eqn_optimal_MSE}
\tdh_{n}=0 \quad , \: n =2,\ldots,N
\end{equation}
and any $\tdh_{1} \ne 0$.
\end{theorem}
%
%
\begin{proof}
See appendix.
\end{proof}
Thm.~\ref{thm_optimal_MSE} is in agreement with intuition: the optimal design is the one that lets the mean frequency component unfiltered while suppressing the rest of the frequencies (i.e. an ideal low-pass graph filter), yielding an output that contains only the desired value.

Interestingly enough, a similar solution also minimizes the volume of the ellipsoid given by the covariance matrix subject to an energy constraint.

%
\begin{theorem}
	\label{thm_optimal_logdet}
Under the conditions of Thm.~\ref{thm_optimal_MSE}, the frequency coefficients of the filter taps that solve
\begin{align}
\underset{\tdh_{1},\ldots,\tdh_{N}}{\text{minimize }} 
	& \log(\det(\bbC_{z}))
	 	\label{eqn_logdet} \\
\text{subject to }
	& \|\tbh\|_{2}^{2} \le \nu_{\max}
\end{align}
are also given by \eqref{eqn_optimal_MSE} and $\tdh_{1} = \sqrt{\nu_{\max}}$ for some $\nu_{\max} > 0$ that determines the maximum energy in the GFT of the filter taps.
\end{theorem}
%
\begin{proof}
See appendix.
\end{proof}
\noindent Theorem \ref{thm_optimal_logdet} shows that graph filter \eqref{eqn_optimal_unbiased} is also optimal in the sense that it minimizes the volume of the error ellipsoid (also known as D-optimality criteria in optimal design of experiments \cite{Pukelsheim93-ExperimentDesign}).

\begin{remark}[Optimal estimators] \normalfont
    \label{rmk_optimal}
We observe that the estimator \eqref{eqn_optimal_unbiased} that minimizes the MSE is equal to $\bbv_{1}^{\Hr} \bbx$ (see Theorem~\ref{thm_optimal_MSE}) which is the optimal estimator across all possible linear operators. This implies that the optimal linear shift-invariant graph filter \eqref{eqn_optimal_unbiased} with coefficients \eqref{eqn_optimal_MSE} is not only the optimal estimator among the class of linear shift-invariant graph filters, but also the optimal estimator among the class of all linear operators. We note that using the optimal estimator in the form \eqref{eqn_optimal_unbiased} with coefficients \eqref{eqn_optimal_MSE} offers several advantages. Namely, it can be computed in a decentralized fashion with access to only one node with communication capabilities, whereas using $\bbv_{1}^{\Hr} \bbx$ demands centralized computing since it requires knowledge of the value of the eingenvector $\bbv_{1}$ and the signal $\bbx$ at every node. Additionally, the computational cost of the eigendecomposition is $\ccalO(N^{3})$ whereas \eqref{eqn_optimal_unbiased} demands $\ccalO(N^{2})$ when computed at a single node. In essence, using a linear shift-invariant graph filter \eqref{eqn_optimal_unbiased} with coefficients \eqref{eqn_optimal_MSE} not only guarantees optimality among all linear operators, but also favors a decentralized solution that exploits the sparse and efficient implementation of linear shift-invariant graph filters.
\end{remark}


\begin{figure*}
    \captionsetup[subfigure]{justification=centering}
    \centering
    \begin{subfigure}{0.66\columnwidth}
        \centering
        \includegraphics[width=0.99\textwidth]{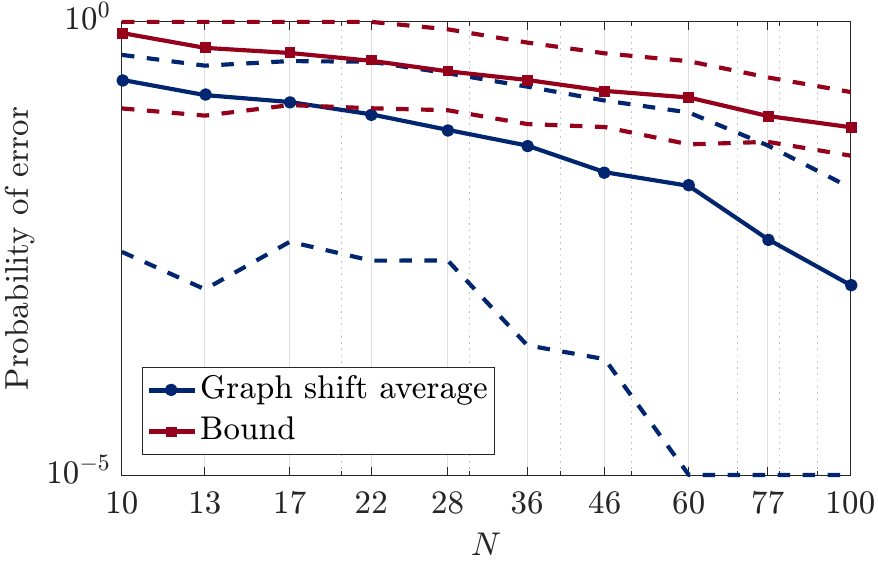}
        \caption{Graph shift average $\hbmu_{N}$ \eqref{eqn_graph_shift_unbiased}}
        \label{p_error_hbmu-COV}
    \end{subfigure}
    \hfill
    \begin{subfigure}{0.66\columnwidth}
        \centering
        \includegraphics[width=0.99\textwidth]{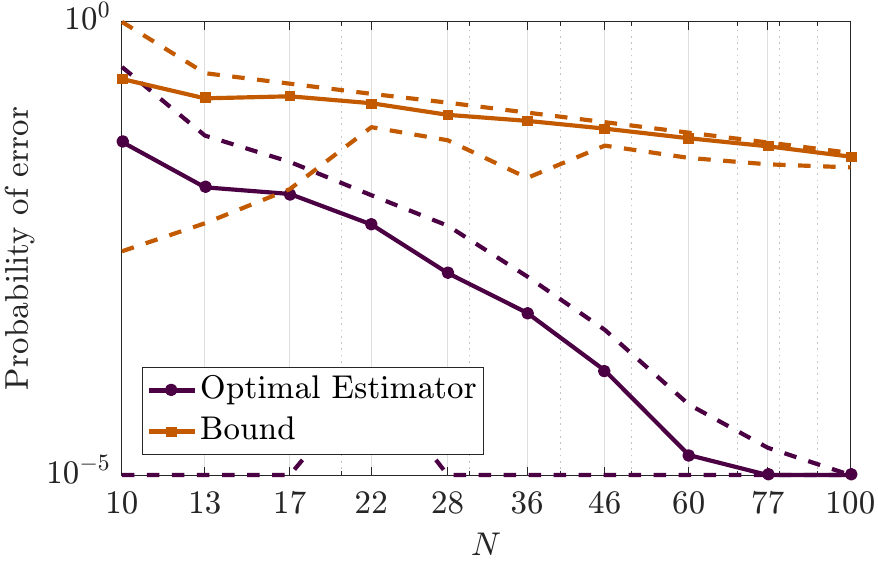}
        \caption{Optimal estimator $\bbz_{N}$ \eqref{eqn_optimal_unbiased}-\eqref{eqn_optimal_MSE}}
        \label{p_error_z-COV}
    \end{subfigure}
    \hfill
    \begin{subfigure}{0.66\columnwidth}
        \centering
        \includegraphics[width=0.99\textwidth]{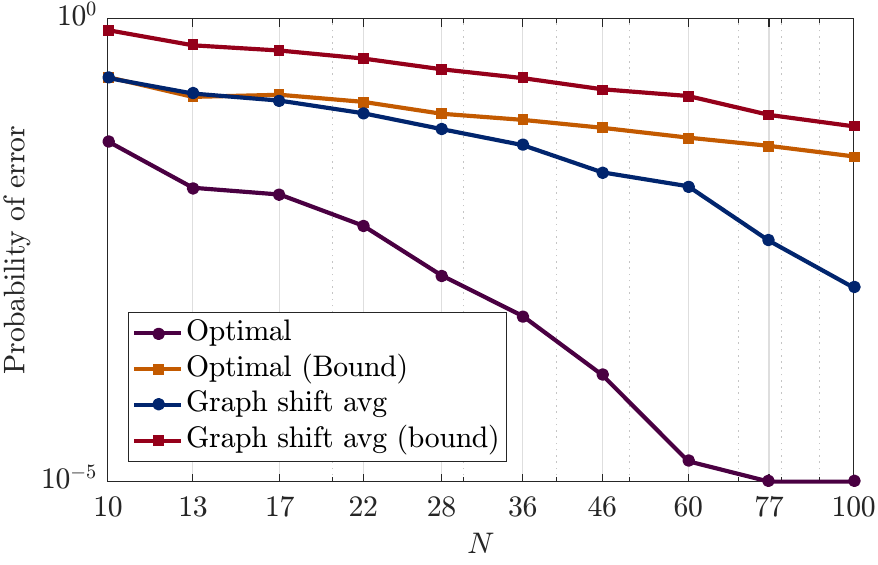}
        \caption{Comparison}
        \label{p_error_comp-COV}
    \end{subfigure}
    \caption{Covariance Graph. The full lines correspond to the mean bound and mean probability of error for $50$ graph realizations. The dashed lines correspond to the maximum and minimum values obtained in some realization. \subref{p_error_hbmu-COV}-\subref{p_error_z-COV} Estimated probability of error and theoretical bound for covariance graphs of varying size $N$ from $10$ to $100$ for the graph shift average and the optimal estimator, respectively. \subref{p_error_comp-COV} Comparison of the mean probability of error and mean theoretical bound for both the graph shift average and the optimal estimator.}
    \label{fig_COV}
\end{figure*}

\section{Numerical Experiments} \label{sec_sims}

In this section we consider numerical experiments to illustrate the effect of the graph shift average \eqref{eqn_graph_shift_unbiased}, the error bound \eqref{eqn_unbiased_error_node} and the optimal estimator \eqref{eqn_optimal_unbiased}-\eqref{eqn_optimal_MSE}. In Sections~\ref{subsec_ER}, \ref{subsec_cov} and \ref{subsec_SBM}, we use these estimators in the context of Erd{\H{o}}s-R{\'e}nyi graphs, covariance graphs and stochastic block models, respectively. In Section~\ref{subsec_GMRF} we deploy this estimator in the context of distributed estimation in sensor networks in which we want to estimate the mean of a Gaussian-Markov Random Field (GMRF). In this last section we compare the performance of the unbiased diffusion estimator with the distributed LMS estimator proposed in \cite{dilorenzo14}.

Unless otherwise specified, we consider an $N$-node graph $\ccalG$ described by a graph shift operator given by the adjacency matrix $\bbS=\bbA=\bbV \bbLambda \bbV^{\Hr} \in \reals^{N \times N}$ which is normal since the adjacency is symmetric (undirected graph). We consider a single realization $\bbx$ of a WSS graph signal with mean $\bbmu=\mu \bbv_{1}$, proportional to the eigenvector associated to the largest eigenvalue. The covariance matrix is determined by the PSD given by  $\bbp \in \reals^{N}$. We define the signal-to-noise ratio as $\textrm{SNR} = 10 \log_{10} (\mu^{2}/p_{1})$ (in dB). For each graph size $N$, $50$ different graphs are generated. For each one of these graphs, $10^{5}$ different signal realizations are simulated and the error probability for some $\epsilon$ is estimated from these, see \eqref{eqn_unbiased_error_node}. Additionally, for each graph realization, we aggregate the results on a node $k$ determined by the node with the largest $v_{1,k}$ such that $v_{1,k}<1/\sqrt{N}$. Results presented include those obtained through averaging across all graphs realizations (full lines), as well as the maximum and minimum results (dashed lines).


\subsection{Erd{\H{o}}s-R{\'e}nyi Graphs} \label{subsec_ER}

In this first example, we consider an $N$-node Erd{\H{o}}s-R{\'e}nyi (ER) graph where each edge is drawn with probability $p_{\ER}=0.2$ independently of all other edges \cite{erdos59}. Only realizations of this graph that are connected are considered. We set $\mu=3$ and $\textrm{SNR}=10 \ \textrm{dB}$. For computing the bound \eqref{eqn_unbiased_error_node} we consider $\epsilon=0.1 \cdot 10^{\textrm{SNR}/10}$. The PSD is given by $p_{n}=p_{1} \cdot 10^{3(n-1)/(N-1)+1}$, i.e. $\bbp$ consists of $n$ logarithmically spaced points between $10 p_{1}$ and $10^{4} p_{1}$.

In the experiment, we vary $N$ from  $10$ to $100$ and also simulate the optimal estimator \eqref{eqn_optimal_unbiased}. In Figs.~\ref{p_error_hbmu-ER} and \ref{p_error_z-ER} we show the error probability and the bound as a function of $N$ for the graph shift average and optimal estimator, respectively. We observe that both decrease as the size of the graph grows larger, as expected from the convergence of Thm.~\ref{thm_glln}. In Fig.~\ref{p_error_comp-ER} we compare the results between both estimators. Note that, because the ER graphs satisfy Thm.~\ref{thm_glln} (see Cor.~\ref{coro_er}), then the graph shift average converges, indeed, to the optimal estimator.

\begin{figure*}
	\captionsetup[subfigure]{justification=centering}
	\centering
	\begin{subfigure}{0.66\columnwidth}
		\centering
		\includegraphics[width=0.99\textwidth]{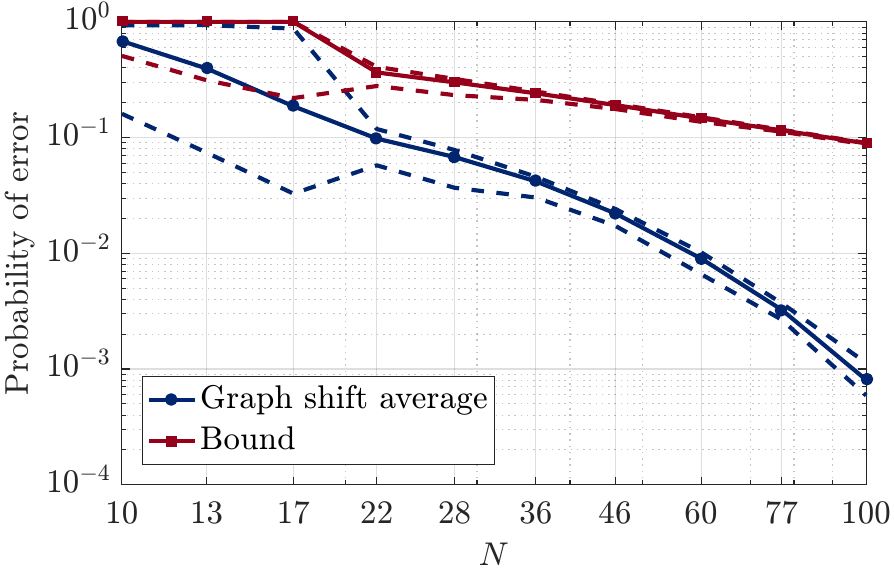}
		\caption{Graph shift average $\hbmu_{N}$ \eqref{eqn_graph_shift_unbiased}}
		\label{p_error_hbmu-SBM}
	\end{subfigure}
	\hfill
	\begin{subfigure}{0.66\columnwidth}
		\centering
		\includegraphics[width=0.99\textwidth]{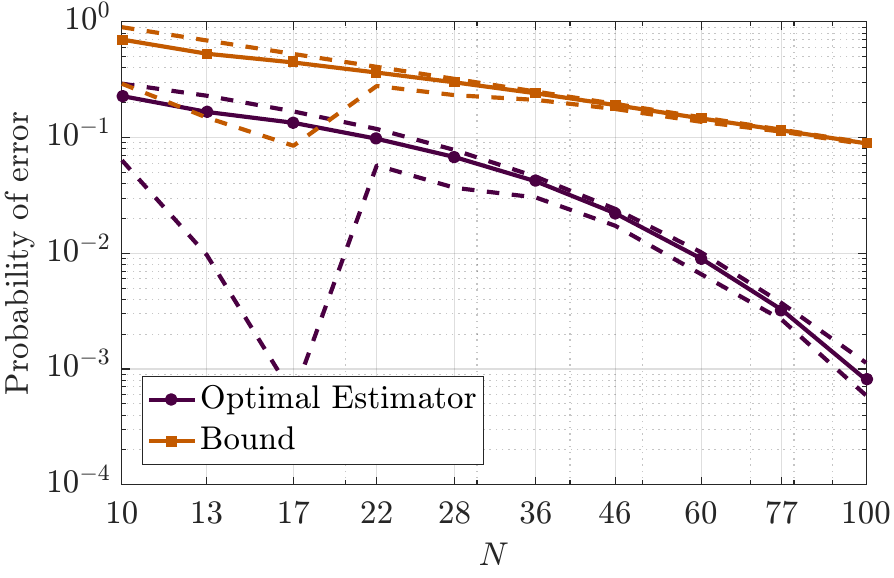}
		\caption{Optimal estimator $\bbz_{N}$ \eqref{eqn_optimal_unbiased}-\eqref{eqn_optimal_MSE}}
		\label{p_error_z-SBM}
	\end{subfigure}
	\hfill
	\begin{subfigure}{0.66\columnwidth}
		\centering
		\includegraphics[width=0.99\textwidth]{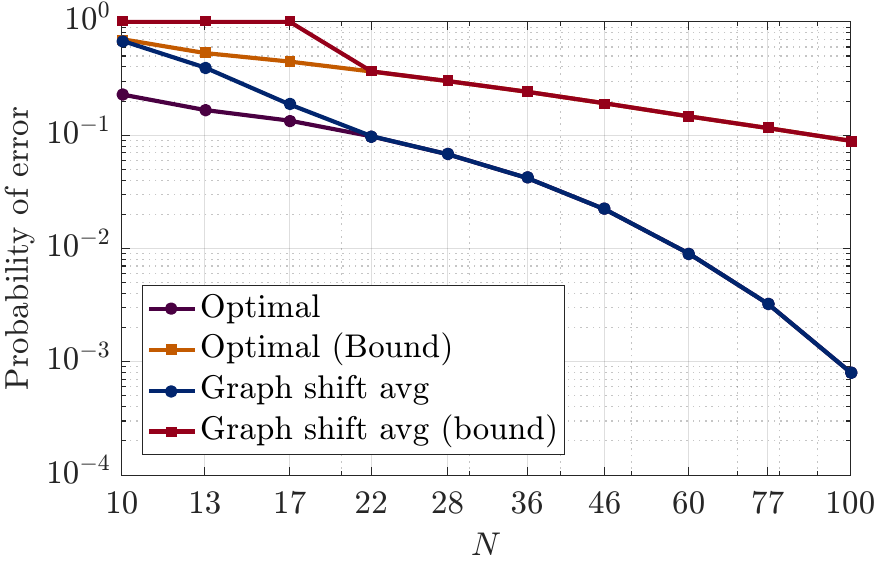}
		\caption{Comparison}
		\label{p_error_comp-SBM}
	\end{subfigure}
	\caption{Stochastic block model. The full lines correspond to the mean bound and mean probability of error for $50$ graph realizations. The dashed lines correspond to the maximum and minimum values obtained in some realization. \subref{p_error_hbmu-SBM}-\subref{p_error_z-SBM} Estimated probability of error and theoretical bound for stochastic block models of varying size $N$ from $10$ to $100$ for the graph shift average and the optimal estimator, respectively. \subref{p_error_comp-SBM} Comparison of the mean probability of error and mean theoretical bound for both the graph shift average and the optimal estimator.}
	\label{fig_SBM}
\end{figure*}


\subsection{Covariance Graphs} \label{subsec_cov}

As a second example, we consider covariance graphs of size $N$. That is, we create a covariance matrix $\bbSigma$ at random of size  $N \times N$, then we generate $10^{5}$ training samples of a zero-mean Gaussian random vector with covariance matrix given by $\bbSigma$. We use these training samples to estimate the covariance matrix and set this estimate $\hbSigma$ as the graph shift operator $\bbS=\hbSigma$. Then, we generate a WSS graph signal over this graph with mean given by $\bbmu = 3 \ \bbv_{1}$ and PSD given by $\bbp=p_{1} \ \bbone$.

For the simulation, we vary $N$ from $10$ to $100$. Estimated error probabilities and bounds can be found in Figs.~\ref{p_error_hbmu-COV} and \ref{p_error_z-COV}. It is observed that both the bound and the estimated error probability decrease with $N$. In Fig.~\ref{p_error_comp-COV} we observe the comparison between the graph shift average and the optimal estimator. Given that a covariance graph does not necessarily satisfy the conditions on Thm.~\ref{thm_glln}, we observe that, while the graph shift average still has decreasing error probability, the optimal estimator does have a better performance having up to $3$ orders of magnitude less error for $N=60$.


\subsection{Stochastic block models} \label{subsec_SBM}

In the third example, we study the performance of the proposed estimators for a stochastic block model (SBM) \cite{decelle11}. A stochastic block model of $N$ nodes with $C$ communities $\{\ccalC_{\alpha},\alpha=1,\ldots,C\}$, $\ccalC_{\alpha} \cap \ccalC_{\beta}= \emptyset$, $\alpha \ne \beta$ and $\cup_{\alpha=1}^{C} \ccalC_{\alpha} = \ccalV$ is constructed in such a way that edges within the same community $\ccalC_{\alpha}$ are drawn independently with probability $p_{\alpha}$ and edges between nodes belonging to different communities $\ccalC_{\alpha}$ and $\ccalC_{\beta}$, $\alpha \ne \beta$, are drawn independently with $p_{\alpha,\beta}$. We consider $C=4$, $p_{\alpha}=0.6$ and $p_{\alpha,\beta}=0.1$ for all $\alpha,\beta=1,\ldots,C$, $\alpha \ne \beta$. For the situations in which $N$ is not divisible by $C$, we add the remainder of the nodes to the communities, one in each, until there are $N$ nodes in the graph. The WSS graph signal is considered to have a PSD given by $p_{n}=p_{1} \cdot 10^{3(n-1)/(N-1)+1}$.

We run simulations for varying graphs sizes $N$ from $10$ to $100$. Figs.~\ref{p_error_hbmu-SBM} and \ref{p_error_z-SBM} show the estimated probability of error and the theoretical bound as a function of $N$. Note that both decrease as $N$ increases. Finally, Fig.~\ref{p_error_comp-SBM} shows the compared error probability between the graph shift average and the optimal estimator. It is observed that for $N>22$ both estimators yield the same result. It is believed that, since the SBM is a combination of ER graphs, then its eigenvalues might satisfy the conditions of Thm.~\ref{thm_glln} and thus the graph shift average is optimal.


\subsection{Gaussian-Markov Random Fields} \label{subsec_GMRF}

\begin{figure*}
	\captionsetup[subfigure]{justification=centering}
	\centering
	\begin{subfigure}{0.95\columnwidth}
		\centering
		\includegraphics[width=0.95\textwidth]{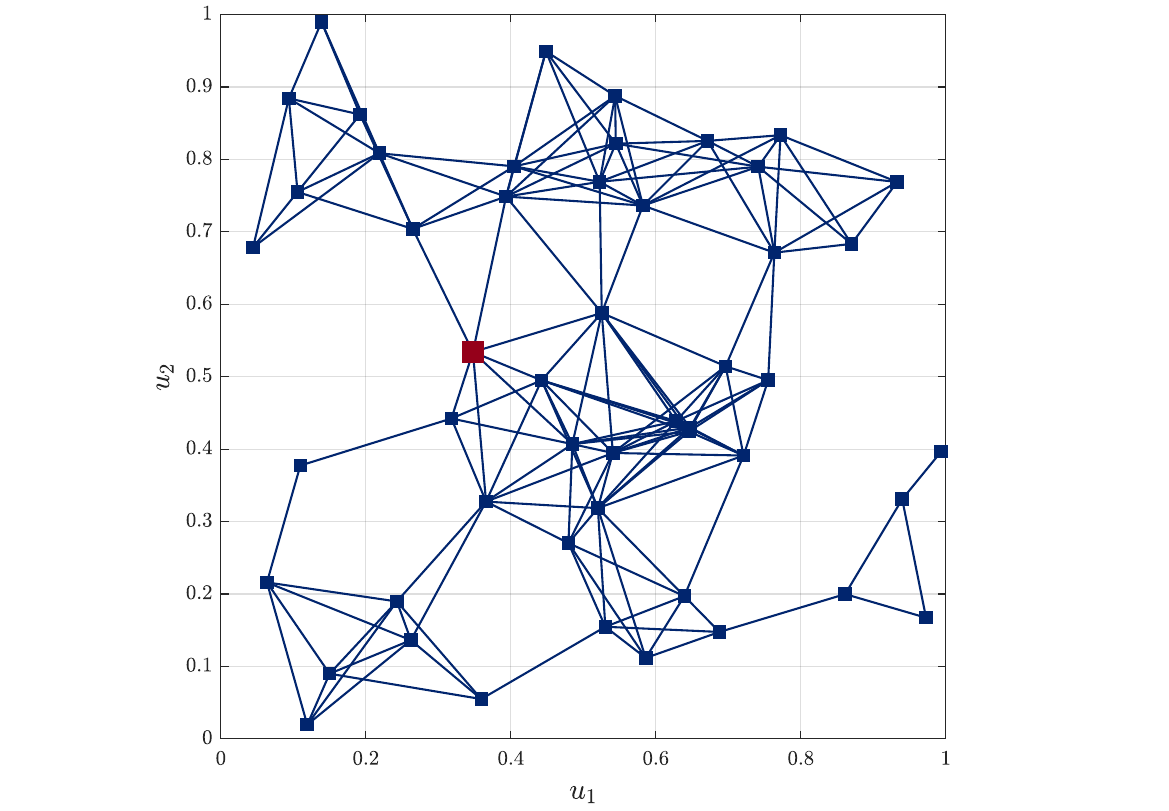}
		\caption{Sensor network}
		\label{gmrf_graph}
	\end{subfigure}
	\hfill
	\begin{subfigure}{0.95\columnwidth}
		\centering
		\includegraphics[width=0.95\textwidth]{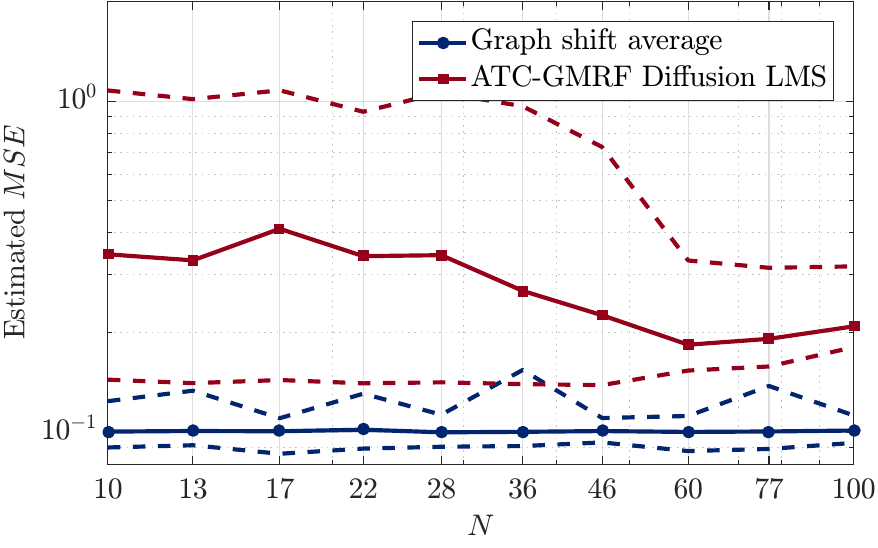}
		\caption{Varying graph size $N$}
		\label{gmrf_error_n}
	\end{subfigure}
	\caption{Gaussian-Markov random field estimation. \subref{gmrf_graph} Example of a $N=50$ sensor network. Sensors are located at random. A weight function $\rho(i,j)=\alpha e^{-\beta \|\bbu_{i}-\bbu_{j}\|_{2}^{2}}$ is computed among nodes. A weighted edge between two nodes $i$ and $j$ is drawn if $\rho(i,j)>\rho_{\textrm{thres}}$ (the weight of the edge is $\rho(i,j)$).  \subref{gmrf_error_n} Estimated MSE for the graph shift average as well as the ATC-GMRF diffusion LMS estimator for a sensor network of varying graph size $N$.}
	\label{fig_GMRF}
\end{figure*}

As a final example we consider the problem of estimating the mean of a Gaussian-Markov random field (GMRF) \cite{rue05}. This problem arises particularly in the context of sensor networks in which measurements are correlated based on the distance among these sensors. Let $\ccalG=(\ccalV,\ccalE,\ccalW)$ be the $N$-sensor network. We consider $N$ sensors deployed at random over a field. The influence between sensors $i$ and $j$ is described by a function $\rho(i,j)=\alpha e^{-\beta \|\bbu_{i}-\bbu_{j}\|_{2}^{2}}$, where $\bbu_{i}$ and $\bbu_{j}$ are $\reals^{2}$ vectors describing the positions of sensor $i$ and $j$ respectively, $i \ne j$, $i,j=1,\ldots,N$. The constants $\alpha$ and $\beta$ are chosen such that $\max_{i,j} \rho(i,j)=\rho_{\max}$ and $\min_{i,j} \rho(i,j) = \rho_{\min}$. We set $\rho(i,i)=0$. An edge between two sensors is drawn whenever the influence function exceeds a certain threshold $(i,j) \in \ccalE \Leftrightarrow \rho(i,j) \ge \rho_{\textrm{thres}}$. The weight function $\ccalW$ is given by the influence function $\ccalW = \rho$; note that we consider graphs without self-loops so that $(i,i) \notin \ccalE$. A WSS graph signal is a GMRF if it has a Gaussian distribution with covariance matrix $\bbC_{x}=|a_{0}|^{2} (\bbI - a \bbS)^{-1} [(\bbI - a \bbS)^{-1}]^{\Hr}$ \cite{marques17}.

To further illustrate this simulation, we revisit Fig.~\ref{fig_gmrf_1000} discussed in the introduction. More precisely, in this scenario, we consider $N=1000$ sensor distributed uniformly at random over the area, each of which takes one measurement. The measurement of each sensor is shown in Fig.~\ref{fig_sensor_measurements}. This realization of the GMRF process is then diffused through the graph to compute the graph shift average $\hbmu_{N}$ as in \eqref{eqn_graph_shift_unbiased}. The resulting value at each sensor is shown in Fig.~\ref{fig_sample_mean}. We can see how the output of the graph shift average is an accurate estimator of the true mean field of the GMRF process, which is shown in Fig.~\ref{fig_true_mean}.

For the other simulations in Fig.~\ref{fig_GMRF}, we set the mean of $\bbx$ to be $\bbmu = \mu \cdot \bbv_{1}$ with $\mu=3$. For building the graph we consider $\rho_{\min}=0.01$, $\rho_{\max}=1$ and $\rho_{\textrm{thres}}=1.75 \cdot \bar{\rho}$ where $\bar{\rho}$ is the average of all the elements in $\{\rho(i,j), i \ne j, i,j=1,\ldots,N\}$. We set $a=1/\lambda_{1}$ and $a_{0}$ so that $\textrm{SNR}=10 \ \textrm{dB}$. For each value of $N$ we simulate $50$ different sensor networks (see Fig.~\ref{gmrf_graph} for an example) and for each network we simulate $10^{3}$ realizations of the WSS graph signal to compute the estimated MSE. In our simulations we compare the MSE resulting from using the graph shift average \eqref{eqn_graph_shift_unbiased} with the MSE obtained from estimating the mean according to the ATC-GMRF diffusion LMS algorithm introduced in \cite{dilorenzo14}.

In Fig.~\ref{gmrf_error_n} we computed the estimated MSE as a function of the graph size $N$. We observe that the graph shift average \eqref{eqn_graph_shift_unbiased} performs better than the ATC-GMRF diffusion LMS. It is worth pointing out that the graph shift average is designed to work on stationary processes whereas the ATC-GMRF diffusion LMS algorithm also works for nonstationary GMRFs.


\section{Conclusions} \label{sec_conclusions}

In the present paper we set to expand the field of statistical graph signal processing by developing a first notion of ergodicity. More precisely, we computed the realization average as a graph shift average: a diffusion of a single realization through the graph. We proved, in a result reminiscent of the WLLN, that this graph shift average converges, under some mild conditions on the graph, to the ensemble mean.

For graphs that do not satisfy the conditions for the WLLN we proposed a LSI graph filter that, when applied to a single realization of the WSS graph signal, yields an unbiased estimator that converges to the ensemble mean on any graph; this is achieved by carefully designing the filter taps so as to account for the specific graph spectrum. Furthermore, the LSI graph filter is optimal in the sense that it minimizes both the mean squared error as well as the volume of the ellipsoid determined by the error covariance matrix.

Finally, we simulated WSS graph signals on several supports to illustrate the theoretical results. We noted that both the probability of error as well as the bound decrease as the size of the graph gets larger. We also observed that for ER graphs and SBMs, the graph shift average and the optimal estimator coincide, and for covariance graphs the optimal estimator yields better results. Additionally, we studied the problem of estimating the mean of a GMRF which typically arises when considering measurements obtained from a sensor networks. We compared the graph shift average with the ATC-GMRF diffusion LMS algorithm and showed better performance of the former.



\appendices


\section{Moments of graph shift averages: Proof of Propositions~\ref{prop_unbiased}~and~\ref{prop_psd_graph_shift_average}.}

%
\begin{proof}[Proof of Proposition~\ref{prop_unbiased}]
Computing the expectation of \eqref{eqn_graph_shift} we obtain
\begin{align}\label{eqn_prop_unbiased_pf_10}
   \mbE [ \hbmu_{N} ]
	    =  \frac{1}{\alpha(\bbS)} 
	       \sum_{\ell=0}^{N-1} \bbS^{\ell} \mbE [\bbx] .
\end{align}
Observe now that since the signal $\bbx$ is stationary on the graph, it holds that its expectation is $\mbE [\bbx] = \mu \bbv_{1}$ for some scalar $\mu$. Substituting this fact into \eqref{eqn_prop_unbiased_pf_10} and reordering terms yields
\begin{align}\label{eqn_prop_unbiased_pf_20}
   \mbE [ \hbmu_{N} ]
	    = \frac{1  }{\alpha(\bbS)} 
	      \sum_{\ell=0}^{N-1} \bbS^{\ell} \mu \bbv_{1}
   	    = \frac{\mu}{\alpha(\bbS)} 
	      \sum_{\ell=0}^{N-1} \bbS^{\ell}     \bbv_{1} .
\end{align}
But since $\bbv_{1}$ is an eigenvector of $\bbS$ associated with eigenvalue $\lam_1$ we have $\bbS^{\ell}  \bbv_{1} = \lambda_{1}^{\ell} \bbv_{1}$ which reduces \eqref{eqn_prop_unbiased_pf_20} to
\begin{align}\label{eqn_prop_unbiased_pf_30}
   \mbE [ \hbmu_{N} ]
   	    = \frac{\mu}{\alpha(\bbS)} 
	      \sum_{\ell=0}^{N-1} \lam_1^{\ell}     \bbv_{1} .
\end{align}
Using the condition in \eqref{eqn_mean_definition} we substitute $\bbmu= \mu\bbv_{1}$ in \eqref{eqn_prop_unbiased_pf_30} and reorder terms to obtain \eqref{eqn_unbiased}. 
\end{proof}

%
\begin{proof}[Proof of Proposition~\ref{prop_psd_graph_shift_average}]
First, we prove that the graph shift average \eqref{eqn_graph_shift_unbiased} is WSS with respect to $\bbS$ [cf. Def.~\ref{def_stationarity}]. It is noted that (i) is satisfied since $\mbE[\hbmu_{N}]=\bbmu=\mu \bbv_{1}$ where $\bbv_{1}$ is an eigenvector of $\bbS$. To prove that (ii) holds, we compute the covariance matrix $\bbC_{\hhatmu}$ and show that it is diagonalizable by $\bbV$. We start by using \eqref{eqn_graph_shift_unbiased} to write
\begin{align}
& \hbmu_{N} - \bbmu 
	 = \frac{1}{\sum_{\ell=0}^{N-1} \lambda_{1}^{\ell}} 
		\sum_{\ell=0}^{N-1} \bbS^{\ell}\bbx
		- \bbmu 
		\label{eqn_hatmu_mu_def} \\
	&= \frac{1}{\sum_{\ell=0}^{N-1} \lambda_{1}^{\ell}} 
		\sum_{\ell=0}^{N-1} \bbS^{\ell}(\bbx-\bbmu) 
	 + \frac{1}{\sum_{\ell=0}^{N-1} \lambda_{1}^{\ell}} 
			\sum_{\ell=0}^{N-1} \bbS^{\ell} \bbmu
		- \bbmu
		\label{eqn_second_term}
\end{align}
Using the fact that $\bbmu = \mu \bbv_{1}$, the second term in \eqref{eqn_second_term} is equivalent to
\begin{equation}
	\frac{1}{\sum_{\ell=0}^{N-1} \lambda_{1}^{\ell}} 
		\sum_{\ell=0}^{N-1} \bbS^{\ell} \mu \bbv_{1} 
	- \bbmu
	= \frac{\mu}{\sum_{\ell=0}^{N-1} \lambda_{1}^{\ell}}
		\sum_{\ell=0}^{N-1} \bbS^{\ell} \bbv_{1}
	- \bbmu.
\end{equation}
Recalling that $\bbS^{\ell}\bbv_{1} = \lambda_{1}^{\ell}\bbv_{1}$ since $\bbv_{1}$ is the eigenvector of $\bbS$ associated to $\lambda_{1}$ and reordering terms we further get
\begin{equation}
	\frac{\mu}{\sum_{\ell=0}^{N-1} \lambda_{1}^{\ell}}
		\sum_{\ell=0}^{N-1} \lambda_{1}^{\ell} \bbv_{1}
	- \bbmu
	=\frac{\sum_{\ell=0}^{N-1} \lambda_{1}^{\ell}}{\sum_{\ell=0}^{N-1} \lambda_{1}^{\ell}} \mu \bbv_{1}
	- \bbmu
	= \bbzero
\end{equation}
where property (i) of WSS graph signals $\bbmu = \mu \bbv_{1}$ was used once more. Canceling out the second term of \eqref{eqn_second_term}, then \eqref{eqn_hatmu_mu_def} yields
\begin{equation} \label{eqn_hatmu_mu}
\hbmu_{N} - \bbmu 
	= \frac{1}{\sum_{\ell=0}^{N-1} \lambda_{1}^{\ell}} 
		\sum_{\ell=0}^{N-1} \bbS^{\ell}(\bbx-\bbmu).
\end{equation}
Eq. \eqref{eqn_hatmu_mu} can be immediately used to compute the covariance matrix $\bbC_{\hhatmu} = \mbE[(\hbmu_{N}-\bbmu)(\hbmu_{N}-\bbmu)^{\Hr}]$ as follows
\begin{equation}
\bbC_{\hhatmu} = 
	\left( \frac{1}{\sum_{\ell=0}^{N-1} \lambda_{1}^{\ell}} 
		\sum_{\ell=0}^{N-1} \bbS^{\ell} \right)
		\bbC_{x}
	\left( \frac{1}{\sum_{\ell=0}^{N-1} \lambda_{1}^{\ell}} 
		\sum_{\ell=0}^{N-1} (\bbS^{\ell})^{\Hr} \right)
\end{equation}
where linearity of the expectation and the fact that $\bbC_{x} = \mbE[(\bbx-\bbmu)(\bbx-\bbmu)^{\Hr}]$ was used. But $\bbx$ is WSS on $\bbS$ and thus $\bbC_{x} = \bbV \diag(\bbp) \bbV^{\Hr}$. This fact, together with $\bbS = \bbV \bbLambda \bbV^{\Hr}$ and reordering terms yields
\begin{equation} \label{eqn_pre_Cmu}
\bbC_{\hhatmu} 
	= \frac{1}{|\sum_{\ell=0}^{N-1} \lambda_{1}^{\ell}|^{2}} 
	\bbV 
		\left( \sum _{\ell=0}^{N-1} \bbLambda^{\ell} \right) 
		\diag(\bbp) 
		\left( \sum _{\ell=0}^{N-1} (\bbLambda^{\ell})^{\Hr} \right)
	\bbV^{\Hr}.
\end{equation}
Since $\bbLambda$ is a diagonal matrix, then \eqref{eqn_pre_Cmu} can be written as
\begin{equation} \label{eqn_Cmu}
C_{\hhatmu} = \bbV \diag(\bbq) \bbV^{\Hr}
\end{equation}
for a PSD vector $\bbq \in \reals^{N}$, proving that (ii) in Def.~\ref{def_stationarity} holds and thus $\hbmu_{N}$ is WSS on $\bbS$. Furthermore, each element $q_{n}$ of the PSD $\bbq$ is given by \eqref{eqn_qk} completing the proof.
\end{proof}


\section{Proving the weak law of large numbers: Proof of Lemmas~\ref{l_unbiased_error_node}~and~\ref{l_behavior_qk}.}

%
\begin{proof}[Proof of Lemma~\ref{l_unbiased_error_node}]
Let $\bbe_{k}$ be a vector containing all zeros except for a $1$ in position $k \in \{1,\ldots,N\}$. Then, we can write $[\hbmu_{N}-\bbmu]_{k} = \bbe_{k}^{\Tr} (\hbmu_{N}-\bbmu)$. Since $\mbE[\bbe_{k}^{\Tr} \hbmu_{N}]=[\bbmu]_{k}$ we have that $\textrm{var}(\bbe_{k}^{\Tr}(\hbmu_{N}-\bbmu)) = \bbe_{k}^{\Tr} \bbC_{\hat{\mu}} \bbe_{k} = \bbe_{k}^{\Tr} \bbV \diag(\bbq) \bbV^{\Hr} \bbe_{k} < \infty$. Noting that $\bbe_{k}^{\Tr} \bbV = [v_{k,1},\ldots,v_{k,N}]$ is the $k$th row of $\bbV$, the variance of $\bbe_{k}^{\Tr}(\hbmu_{N}-\bbmu)$ turns out to be
\begin{equation} \label{eqn_var_hbmu}
\textrm{var}\left( \bbe_{k}^{\Tr}(\hbmu_{N}-\bbmu) \right) = \sum_{n=1}^{N} q_{n} |v_{k,n}|^{2}.
\end{equation}
Finally, Chebyshev's inequality \cite[Thm.~1.6.4]{durrett10} is applied to obtain \eqref{eqn_unbiased_error_node}.
\end{proof}

%
\begin{proof}[Proof of Lemma~\ref{l_behavior_qk}]
Let $\lambda_{n}=R e^{\jj \theta}$ for $n=2,\ldots,N$, and where $\jj$ denotes the imaginary unit $\jj^{2}=-1$. Assume first that $\lambda_{1}>1$ and without loss of generality that $R \neq 1$. Then, using the geometric sum on \eqref{eqn_qk} we have
\begin{equation}
q_{n} 
	= p_{n} 
		\frac{|1-\lambda_{1}|^{2}}{|1-\lambda_{n}|^{2}} 
		\frac{|1-\lambda_{n}^{N}|^{2}}{|1-\lambda_{1}^{N}|^{2}}
	\le
	p_{n} 
		\frac{(1-\lambda_{1})^{2}}{(1-R)^{2}}
		\frac{(1+R^{N})^{2}}{(1-\lambda_{1}^{N})^{2}}
\end{equation}
Now, because $\lambda_{1}>1$ then $|1-\lambda_{1}|^{2}=\ccalO(\lambda_{1}^{2})$ and $|1-\lambda_{1}^{N}|^{2}=\ccalO(\lambda_{1}^{2N})$. Likewise, if $R>1$ then $(1-R)^{2}=\ccalO(R^{2})$ and $(1+R^{N})^{2}=\ccalO(R^{2N})$ so that
\begin{equation}
\frac{(1-\lambda_{1})^{2}}{(1-R)^{2}} \frac{(1+R^{2})^{N}}{(1-\lambda_{1}^{N})^{2}} 
	= \ccalO \left( \frac{\lambda_{1}^{2}}{R^{2}} \frac{R^{2N}}{\lambda_{1}^{2N}} \right) 
	= \ccalO \left( \frac{R^{2(N-1)}}{\lambda_{1}^{2(N-1)}} \right).
\end{equation}
And, by hypothesis, $|\lambda_{n}|/\lambda_{1} = R/\lambda_{1} = o(N^{-\delta/2(N-1)})$ so that
\begin{equation}
\frac{(1-\lambda_{1})^{2}}{(1-R)^{2}} \frac{(1+R^{2})^{N}}{(1-\lambda_{1}^{N})^{2}} 
	= \ccalO \left( \frac{|\lambda_{n}|^{2(N-1)}}{\lambda_{1}^{2(N-1)}} \right) = o(N^{-\delta}).
\end{equation}
If $R < 1$, and since $\lambda_{1}>1$, then
\begin{equation}
\frac{(1-\lambda_{1})^{2}}{(1-R)^{2}}
		\frac{(1+R^{N})^{2}}{(1-\lambda_{1}^{N})^{2}} 
	= \ccalO \left( \frac{1}{\lambda_{1}^{2(N-1)}} \right) 
	= o (N^{-\delta}).
\end{equation}

For $\lambda_{1}=1$, we have that $|\sum_{\ell=0}^{N-1} \lambda_{1}^{\ell}|^{2} = N^{2}$ so that
\begin{equation}
q_{n}
	= \frac{p_{n}}{N^{2}} \left| \sum_{\ell=0}^{N-1} \lambda_{n}^{\ell} \right|^{2}
	= \frac{p_{n}}{N^{2}} \frac{|1-R^{N} e^{\jj N\theta}|^{2}}{|1-\lambda_{n}|^{2}}.
\end{equation}
For $R<\lambda_{1}=1$ we have that $|1-R^{N} e^{\jj \theta N}|^{2}=\ccalO(1)$. If $R=o(1)$, then $|1-Re^{\jj \theta}|^{2}=\ccalO(1)$ and the $1/N^{2}$ guarantees that $q_{n}=o(1/N)$. If $R=1$ then $|1-e^{\jj N\theta}|^{2}/|1-e^{\jj \theta}|^{2}$ oscillates in a bounded fashion so that, again, the factor $1/N^{2}$ guarantees $q_{n}=o(1/N)$ completing the proof.
\end{proof}


\section{Special cases: Proof of Corollaries~\ref{coro_wlln}~and~\ref{coro_er}.}

%
\begin{proof}[Proof of Corollary~\ref{coro_wlln}]
Note that $\lambda_{1}=1$ so that Thm.~\ref{thm_glln} holds. More specifically, $q_{n}=0$ for all $n \ne 1$ and $v_{k,1}=1/\sqrt{N}$ for all $k \in \{1,\ldots,N\}$ so that $\bbmu=\mu \bbone$ is the constant vector. Finally, $\sum_{\ell=0}^{N-1}\bbS^{\ell} \bbx = \sum_{n=1}^{N}x_{n} \bbone_{N}$ because after $N$ shifts the values of the signal have been aggregated at all nodes due to the nature of the directed cycle, see Fig.~\ref{fig_dc}.
\end{proof}

%
\begin{proof}[Proof of Corollary~\ref{coro_er}]
First, note that $\lambda_{1}=Np_{\ER}+o(N)$ and that by the semi-circle law, with probability $1-o(1)$ all eigenvalues except the largest one lie in the interval $(-c\sqrt{N},c\sqrt{N})$ for any $c>2p_{\ER}(1-p_{\ER})$ \cite{furedi81, krivelevich03}. Then, we have that $\lambda_{2} \le c\sqrt{N}$ so that $\lambda_{2}/\lambda_{1}=o(N^{-\delta/2(N-1)})$ for any $0<\delta<N-1$, satisfying Thm.~\ref{thm_glln}. Additionally, because $\sqrt{N}v_{k,1} = 1+o(N^{-(1/2-r)})$, $0<r<1/2$ with probability $1-o(1)$, see \cite{mitra09}, then any node $k \in \{1,\ldots,N\}$ yields similar probability of error.
\end{proof}


\section{Optimal mean estimation: Proof of Lemma~\ref{l_nonconvergent},~Propositions~\ref{prop_optimal_unbiased}~and~\ref{prop_psd_optimal}~and~Theorems~\ref{thm_optimal_MSE}~and~\ref{thm_optimal_logdet}.}

%
\begin{proof}[Proof of Lemma~\ref{l_nonconvergent}]
In analogy with the proof of Lemma~\ref{l_behavior_qk} we prove that for the conditions of $\lambda_{1}<1$ for which $\ccalM=\{2,\ldots,N\}$ or for the case when $\lambda_{1}>1$ and $\ccalM$ is nonempty, then
\begin{equation}
q_{n}=p_{n} \left| \frac{1-\lambda_{1}}{1-\lambda_{n}} \right|^{2} (1+o(1))
	\label{eqn_nonconvergent_qk}
\end{equation}
First, let $m \in \ccalM$ and $\lambda_{m}=Re^{\jj\theta}$ with $\lambda_{1}>1$ and $\jj^{2}=-1$. Then, since $R/\lambda_{1}$ does not decrease any faster than $N^{-\delta/2(N-1)}$, we have
\begin{equation}
\frac{|1-\lambda_{m}^{N}|^{2}}{|1-\lambda_{1}^{N}|^{2}}
	= \frac{1-2R^{N} \cos(N\theta) + R^{2N}}{1-2\lambda_{1}^{N} + \lambda_{1}^{2N}} = 1+o(1).
\end{equation}
For $\lambda_{1}<1$ we have that, since $R/\lambda_{1} \le 1$ and $R^{N}=o(1)$ and $\lambda_{1}^{N}=o(1)$, then $|1-R^{N} e^{\jj\theta N}|^{2}=1+o(1)$ and $|1-\lambda_{1}^{N}|^{2}=1+o(1)$, completing the proof.
\end{proof}

%
\begin{proof}[Proof of Proposition~\ref{prop_optimal_unbiased}]
Let us start by computing the expectation of \eqref{eqn_optimal_unbiased}
\begin{equation}
\mbE [ \bbz_{N} ]
	= \frac{1}{\sum_{\ell=0}^{N-1} h_{\ell} \lambda_{1}^{\ell}}
		\sum_{\ell=0}^{N-1} h_{\ell} \bbS^{\ell} \mbE [\bbx].
\end{equation}
Now, by definition of WSS graph signals (Def.~\ref{def_stationarity}), it holds that $\mbE[\bbx]=\mu \bbv_{1}$ where $\bbv_{1}$ is the eigenvector associated to eigenvalue $\lambda_{1}$ so that $\bbS^{\ell}\bbv_{1} = \lambda_{1}^{\ell} \bbv_{1}$. Then,
\begin{equation}
\mbE[\bbz_{N}]
	= \frac{1}{\sum_{\ell=0}^{N-1} h_{\ell} \lambda_{1}^{\ell}} 
		\sum_{\ell=0}^{N-1} \mu \ h_{\ell}\bbS^{\ell} \bbv_{1}
	= \mu \cdot 
		\frac{\sum_{\ell=0}^{N-1}h_{\ell} \lambda_{1}^{\ell} \bbv_{1}}{\sum_{\ell=0}^{N-1}h_{\ell} \lambda_{1}^{\ell}}.
\end{equation}
Finally, noting that $\bbv_{1}$ does not depend on the index of the summation, and taking it out, we observe that numerator and denominator are the same, and thus,
\begin{equation}
\mbE[\bbz_{N}]
	= \mu \bbv_{1}\cdot 
		\frac{\sum_{\ell=0}^{N-1}h_{\ell} \lambda_{1}^{\ell}}{\sum_{\ell=0}^{N-1}h_{\ell} \lambda_{1}^{\ell}}
	= \mu \bbv_{1} = \bbmu.
\end{equation}
which completes the proof.
\end{proof}

%
\begin{proof}[Proof of Proposition~\ref{prop_psd_optimal}]
The unbiased LSI graph filter \eqref{eqn_optimal_unbiased} has GFT coefficients given by $\tdh_{n}/\tdh_{1}$ [cf.~\eqref{eqn_GFT_filter}]. Then, from \cite[Property 1]{marques17} it is obtained that each element of the PSD of the output $r_{n}$ of a LSI graph filter is equal to the squared magnitude of each frequency coefficient of the filter $|\tdh_{n}|^{2}/|\tdh_{1}|^{2}$, multiplied by the corresponding PSD coefficient of the input $p_{n}$. This yields \eqref{eqn_psd_optimal}. The expression for the covariance matrix $\bbC_{z}$ readily follows, cf.~Sec.~\ref{sec_prelim}.
\end{proof}

%
\begin{proof}[Proof of Theorem~\ref{thm_optimal_MSE}]
Start by taking the derivative of \eqref{eqn_MSE} with respect to each frequency coefficient and set it to zero
\begin{align}
\frac{\partial \ \tr[\bbC_{z}]}{\partial h_{1}} 
	&= -2 \frac{\tdh_{1}}{|\tdh_{1}|^{2}}\sum_{n=2}^{N} p_{n} \frac{|\tdh_{n}|^{2}}{|\tdh_{1}|^{2}} = 0;
		\\
\frac{\partial \ \tr[\bbC_{z}]}{\partial h_{n}} 
	&= 2 p_{n} \frac{\tdh_{n}}{|\tdh_{1}|^{2}} = 0 \ , \ n=2,\ldots,N.
\end{align}
Note that by setting $\tdh_{1} \ne 0$ and $\tdh_{n}=0$ for all $n=2,\ldots,N$ both necessary and sufficient conditions are satisfied, and thus, these are the optimal frequency coefficients of the filter.
\end{proof}

%
\begin{proof}[Proof of Theorem~\ref{thm_optimal_logdet}]
Recall that the determinant is the product of the eigenvalues so that
\begin{equation}
\log(\det(\bbC_{z}))
	= \log \left( \prod_{n=1}^{N} r_{n} \right)
	= \sum_{n=1}^{N} \log \left( p_{n} \frac{|\tdh_{n}|^{2}}{|\tdh_{1}|^{2}} \right).
\end{equation}
This, in turn, can be rewritten as
\begin{equation}
\log(\det(\bbC_{z}))
	=\sum_{n=1}^{N} \left[ \log( p_{n}) + 2 \log (|\tdh_{n}|) - 2 \log(|\tdh_{1}|) \right].
\end{equation}
Note that $\log(|\tdh_{n}|)$ is minimized when $\tdh_{n}=0$. Also, $\log(|\tdh_{1}|)$ is a monotone increasing function so that $-\log(|\tdh_{1}|)$ is minimized for the largest possible value of $\tdh_{1}$. Then, by setting $\tdh_{n}=0$ for all $n=2,\ldots,N$ and $\tdh_{1}=\sqrt{\nu_{\max}}$ so that the constraint is satisfied for the largest possible value of $\tdh_{1}$ we effectively minimize the objective function, thus completing the proof.
\end{proof}

\bibliographystyle{IEEEtran}
\bibliography{myIEEEabrv,biblioGLLN}

\end{document}